\documentclass[lettersize,journal]{IEEEtran}
\usepackage{colortbl}
\usepackage[table]{xcolor}
\usepackage{array}
\definecolor{mygray}{gray}{.9}
\usepackage{amsmath,amsfonts}
\usepackage{array}
\usepackage{textcomp}
\usepackage{stfloats}
\usepackage{url}
\usepackage{verbatim}
\usepackage{graphicx}
\usepackage{cite}
\usepackage{epsfig}
\usepackage{array}
\usepackage{xcolor}
\usepackage{url}
\usepackage{bm}
\usepackage{tabularx}
\usepackage{color}
\usepackage{balance}

\usepackage{booktabs}       
\usepackage{threeparttable}
\usepackage{tikz}
\usepackage{epstopdf}
\usepackage{lineno}

\usepackage{times}
\usepackage{epsfig}
\usepackage{amsmath}
\usepackage{amssymb}
\usepackage{helvet}
\usepackage{courier}
\usepackage{multirow}
\usepackage{amsopn}
\usepackage{graphicx,subfigure}
\usepackage{array}
\usepackage{url}
\usepackage{setspace}
\usepackage{bm}
\usepackage{tabularx}
\usepackage[ruled,vlined]{algorithm2e}
\usepackage[figuresright]{rotating}
\usepackage{indentfirst}
\usepackage{epstopdf}
\usepackage{tabularx}
\usepackage{threeparttable}

\usepackage{times}  
\usepackage{helvet}  
\usepackage{courier}  
\usepackage{url}  
\usepackage{epsfig}
\usepackage{amsmath}
\usepackage{amssymb}
\usepackage{amsthm}
\usepackage{multirow}
\usepackage{amsopn}
\usepackage{colortbl}
\usepackage{setspace}
\usepackage{bm}
\usepackage{tabularx}
\usepackage[figuresright]{rotating}
\usepackage{indentfirst}
\usepackage{epstopdf}

\usepackage{booktabs}
\usepackage{multirow}
\usepackage{bm}
\usepackage{arydshln}
\usepackage{float}
\usepackage{csquotes}
\usepackage{makecell}

\usepackage{booktabs}       
\usepackage{microtype}      
\usepackage{algpseudocode}  
\usepackage{threeparttable}

\newtheorem{prop}{Proposition}

\begin{document}

\title{Multi-view Multi-label Fine-grained Emotion Decoding from Human Brain Activity}

\author{Kaicheng~Fu$^\dagger$, Changde~Du$^\dagger$, Shengpei~Wang
	and Huiguang~He$^*$,~\IEEEmembership{Senior Member,~IEEE}
	\IEEEcompsocitemizethanks{
		\IEEEcompsocthanksitem K. Fu is with the Research Center for Brain-Inspired Intelligence, National Laboratory of Pattern Recognition, Institute of Automation, Chinese Academy of Sciences, Beijing 100190, China, also with the School of Artificial Intelligence, University of Chinese Academy of Sciences, Beijing 100049, China (e-mail: fukaicheng2019@ia.ac.cn).
		\IEEEcompsocthanksitem C. Du, S. Wang and H. He are with the Research Center for Brain-Inspired Intelligence, National Laboratory of Pattern Recognition, Institute of Automation, Chinese Academy of Sciences, Beijing 100190, China (e-mail:
		changde.du@ia.ac.cn; wangshengpei2014@ia.ac.cn). H. He is also with the School of Artificial Intelligence, University of Chinese Academy of Sciences, Beijing 100049, China, and the Center for Excellence in Brain Science and Intelligence Technology, Chinese Academy of Sciences, Beijing 100190, China. 
	}
	\thanks{$\dagger$: equal contribution}
	\thanks{ $*$: corresponding author (huiguang.he@ia.ac.cn)}}

\markboth{Journal of \LaTeX\ Class Files,~Vol.~14, No.~8, August~2021}%
{Shell \MakeLowercase{\textit{et al.}}: A Sample Article Using IEEEtran.cls for IEEE Journals}


\maketitle

\begin{abstract}
Decoding emotional states from human brain activity plays an important role in brain-computer interfaces. Existing emotion decoding methods still have two main limitations: one is only decoding a single emotion category from a brain activity pattern and the decoded emotion categories are coarse-grained, which is inconsistent with the complex emotional expression of human; the other is ignoring the discrepancy of emotion expression between the left and right hemispheres of human brain. In this paper, we propose a novel multi-view multi-label hybrid model for fine-grained emotion decoding (up to 80 emotion categories) which can learn the expressive neural representations and predicting multiple emotional states simultaneously. Specifically, the generative component of our hybrid model is parametrized by a multi-view variational auto-encoder, in which we regard the brain activity of left and right hemispheres and their difference as three distinct views, and use the product of expert mechanism in its inference network. The discriminative component of our hybrid model is implemented by a multi-label classification network with an asymmetric focal loss. For more accurate emotion decoding, we first adopt a label-aware module for emotion-specific neural representations learning and then model the dependency of emotional states by a masked self-attention mechanism. Extensive experiments on two visually evoked emotional datasets show the superiority of our method.
\end{abstract}


\section{Introduction}
\IEEEPARstart{E}{motion} decoding from visually evoked human brain activity measured by functional Magnetic Resonance Imaging (fMRI) is an emerging research area and plays an important role in Brain-Computer Interfaces (BCIs). Most of the previous emotion decoding studies \cite{li2018hierarchical,du2018semi,li2020novel} considered very coarse-grained and limited emotion categories, such as positive, neutral and negative, which can not represent the complex emotions we experience in daily life. Recent studies \cite{cowen2017self,koide2020distinct} have proven that the emotions expressed by human constitute a high-dimensional semantic space which motivates us to study more fine-grained emotion decoding. In addition, a more obvious limitation in previous decoding methods is that they only predicted one emotion category for one brain activity pattern \cite{du2018semi,li2020novel,horikawa2020neural}, which means they only regard emotion decoding as a multi-class, rather than multi-label, decoding problem. However, there is no doubt that multiple human emotions can be elicited simultaneously by emotion stimuli. Fig. \ref{fig:video} shows three video screenshots of emotion-inducing movie clips used in our experimental dataset. Taking the third scene fighting as an example, we may feel \emph{angry} and \emph{disgust} for this uncivilized behavior while we are also \emph{surprised} and \emph{confused} with the reason for fighting. Besides, it is \emph{fearful} and \emph{anxious} for us when we see people fighting on the street, which will make us \emph{sad} about the injury as well. 
\begin{figure}
	\centering
	\subfigure[\emph{Adoration}, \emph{Awe}, \emph{Joy}, \emph{Nostalgia} and \emph{Romance}]{\includegraphics[height=0.7in,width=3.3in]{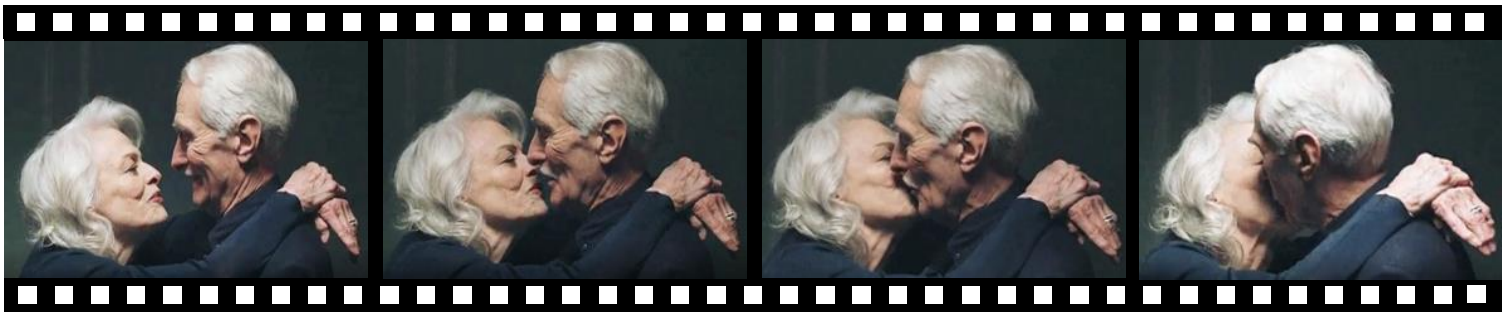}}
	\subfigure[\emph{Adoration}, \emph{Amusement}, \emph{Awe}, \emph{Interest} and \emph{Joy}]{\includegraphics[height=0.7in,width=3.3in]{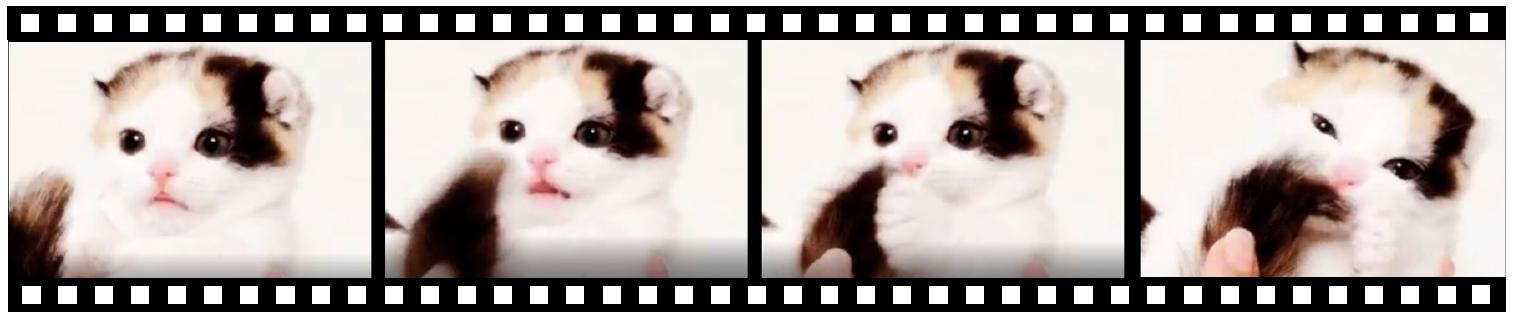}}
	\subfigure[\emph{Anger}, \emph{Anxiety}, \emph{Confusion}, \emph{Disgust}, \emph{Fear}, \emph{Sadness} and \emph{Surprise}.]{\includegraphics[height=0.7in,width=3.3in]{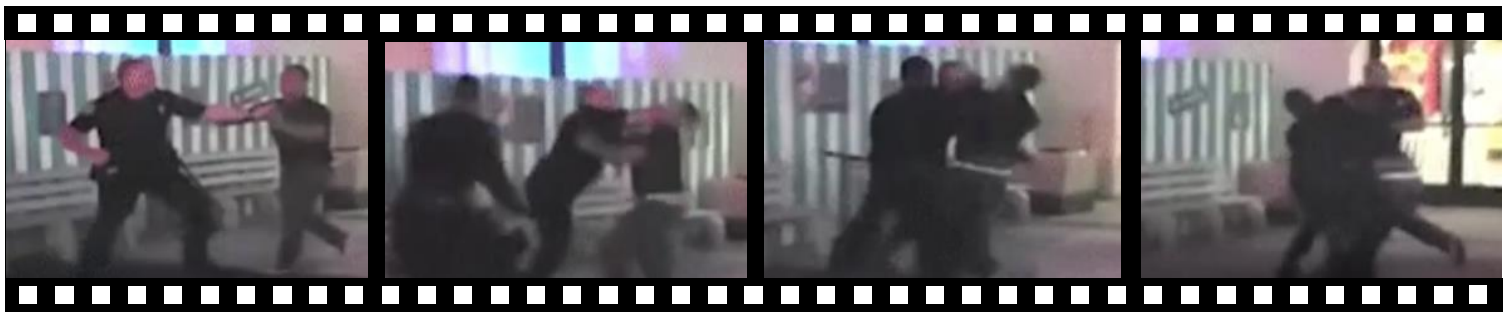}}
	\caption{Several video screenshots and their corresponding emotion categories which illustrate that multiple human emotions can be elicited simultaneously, in which (a) is a romantic scene, (b) is a cat scene and (c) is a fighting scene.}
	\label{fig:video}
	\vspace{-0.5cm}
\end{figure} 

On the other hand, most previous emotion decoding methods adopted voxel based decoding \cite{kragel2016decoding,horikawa2020neural,putkinen2021decoding} which suffers from severe overfitting due to high data dimensionality, small sample size and low signal-to-noise ratio. In order to alleviate overfitting, they barely used simple linear models \cite{horikawa2020neural,saarimaki2021classification} or generalized linear models (GLMs) \cite{putkinen2021decoding} for decoding, leading to weak expressive ability of the model. Although some of them have involved the use of brain regions of interest (ROI) \cite{horikawa2020neural} or functional cortical \cite{putkinen2021decoding} signals for emotion decoding, they  essentially are still subject to voxel based decoding in which advanced deep learning algorithms are not easy to use.

More importantly, for visually evoked emotion decoding, neuroscience research studies have revealed the discrepancy in emotional expression between the left and right hemispheres despite the common visual input \cite{dimond1976differing,davidson1992anterior} which can be regarded as the prior knowledge for developing models. Many previous studies have proven that exploiting the bi-hemisphere discrepancy is capable of improving the performance of emotion decoding by using electroencephalogram (EEG) based data \cite{li2020novel,huang2021differences}. They employ independent networks to learn features for the left and right hemispheres, and then fuse them for further decoding task. Inspired by these, we regard the left and right hemispheres, and their difference as multiple information resources (views) to explore the interaction between the two hemispheres and take advantage of a Product of Experts (PoE) mechanism \cite{cao2014generalized} for feature level fusion.

Taking the above limitations into account, we propose a novel multi-view multi-label hybrid model called ML-BVAE (a hybrid model of \textbf{M}ulti-\textbf{L}abel classification network with \textbf{B}rain \textbf{V}ariational \textbf{A}uto-\textbf{E}ncoder) which is, to the best of our knowledge, the first algorithm for fine-grained emotion decoding from human brain activity and can offer great potential for a fine-grained emotion BCI system. To this end, we first propose an fMRI dimensionality reduction method utilizing ROI pooling which can effectively realize noise suppression and dimensionality reduction of the original signal. Then, we raise a multi-view multi-label hybrid model which enables us to not only achieve accurate emotion decoding but also investigate the relationship between emotion categories. An overview of our experimental paradigm can be found in Fig. \ref{fig:overview}. Our main contributions can be summarized as follows: 
\begin{itemize}
	\item We propose an effective fMRI dimensionality reduction method called ROI pooling which can suppress noise and alleviate overfitting when complex models are used for emotion decoding.
	\item We design a tailored BVAE with the PoE mechanism for learning expressive multi-view neural representations while taking the bi-hemisphere discrepancy into account, which is of great importance to emotion decoding studies.
	\item  We develop a novel multi-label classification network to exploit the relationship between brain activity and emotion labels, and the dependency across emotion labels. For the first time, we realize multi-label emotion decoding from brain activity.
	\item Sufficient experiments on two visually evoked emotional datasets demonstrate that ML-BVAE outperforms the compared methods on several evaluation metrics and can achieve fine-grained emotion decoding of up to 80 emotions. Our source code is available in \url{https://github.com/KaichengFu1997/ML-BVAE}.
\end{itemize}
\begin{figure}
	\centering
	\includegraphics[height=2.1in,width=3.6in]{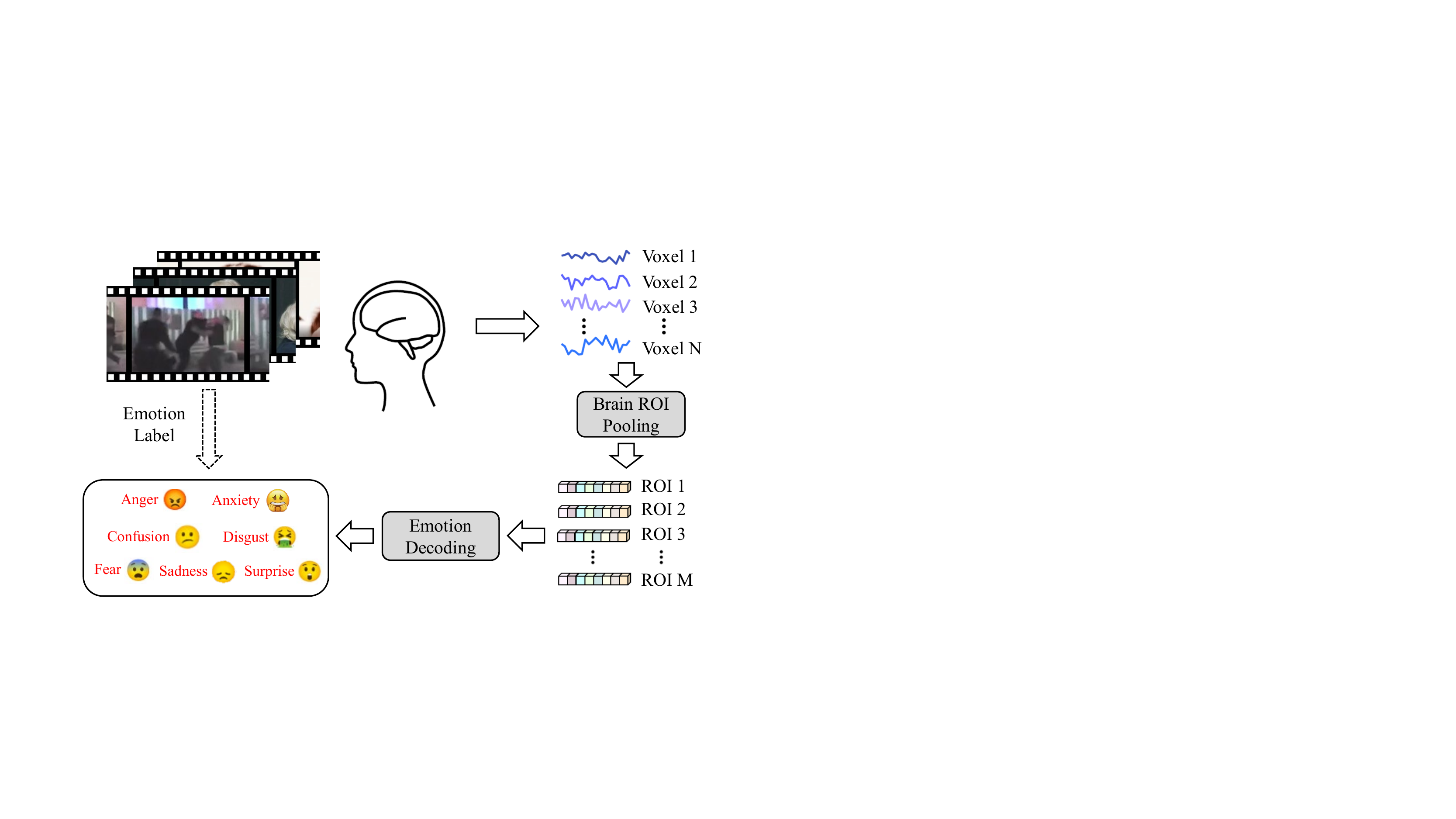}
	\caption{An overview of our experimental paradigm. Subjects watched emotion-inducing movie clips and their fMRI responses were recorded simultaneously. We first convert the voxel data into ROI pooling features, and then perform fine-grained multi-label emotion decoding.}
	\label{fig:overview}
	\vspace{-0.3cm}
\end{figure}  
\section{Related work}
\subsection{Emotion decoding from brain activity} In recent years, decoding emotional states from brain activity has been widely studied. Most existing studies performed emotion decoding based on either fMRI or EEG signals. In fMRI-based emotion decoding, \cite{baucom2012decoding} decoded four affective dimensions with logistic regression with percent signal change features; \cite{horikawa2020neural} used linear regression to make predictions separately for each of the 34 emotion categories and each of the 14 affective dimensions based on the responses of ROI voxels; \cite{putkinen2021decoding} utilized support vector machine (SVM) to build the relationship between four emotion categories and brain activity in the auditory and motor cortex. Although they have made some progress in emotion decoding, they can only use simple linear models limited by high dimensional features. In EEG-based emotion decoding, \cite{song2018eeg} used a dynamical Graph Convolutional Networks (GCNs) to model the multichannel EEG features and performed EEG emotion classification; \cite{shen2021contrastive,zhao2021plug} used contrastive learning and domain adaptation respectively for cross-subject EEG-based emotion recognition; \cite{li2020novel,huang2021differences} both took the discrepancy of bi-hemisphere into consideration for emotion decoding from EEG-based emotional datasets. Despite the complex models and scenarios, they still performed multi-class emotion decoding with coarse-grained emotion categories. Compared with EEG-based methods, fMRI-based studies can perform more fine-grained emotion decoding due to its higher spatial resolution but easily suffers from overfitting. Our method belongs to fMRI-based emotion decoding and we perform dimensionality reduction to the original fMRI data which can alleviate overfitting when complex models such as Deep Neural Networks (DNNs) are used. We also achieve multi-label emotion decoding which means decoding a set of emotional states simultaneously.
\subsection{Multi-label emotion classification} Recently, researchers have noticed the complexity of human emotion expression and have committed to studying the problem of multi-label emotion classification. \cite{fei2020latent} proposed latent emotion memory (LEM) for learning latent emotion distribution and utilized bi-directional GRU to learn emotion coherence so as to realize the multi-label emotion classification of text data. \cite{wang2019capturing} employed adversarial learning and combined adversarial loss and multi-label supervised loss to achieve multi-label emotion tagging for video data. \cite{zhang2021multi2} exploited Graph Neural Networks (GNNs) with heterogeneous hierarchical message passing for multi-modal multi-label emotion classification with textual, visual and acoustic modalities. However, these studies all belong to coarse-grained emotion classification which involve only 12 emotion categories at most \cite{zhang2021multi2}. It is difficult for these studies to truly reflect the complex emotion categories of human. Based on this disadvantage, \cite{abdul2017emonet} built a fine-grained emotion dataset with 24 types of emotions using a large collection of tweets; \cite{demszky2020goemotions} collected another fine-grained emotion dataset named GoEmotions with 58k English Reddit comments, which involves 27 emotion categories. Furthermore, \cite{huang2021seq2emo} proposed a model called Seq2Emo with a bi-directional decoder for multi-label emotion classification on GoEmotions. However, the abovementioned studies are all based on multimedia data, rather than recognizing human emotions based on physiological signals which is also known as emotion decoding. Different from these studies, for the first time, we achieve fine-grained emotion decoding from human brain activity with fMRI data, which can decode up to 80 emotion categories.
\subsection{Multi-view multi-label learning} Multi-label learning has many applications in text categorization \cite{guo2021label}, bioinformatics \cite{zhou2020iatc}, web mining \cite{tang2009large}, etc. In recent years, multi-label learning has also been used in semantic decoding from brain activity \cite{huth2012continuous,huth2016decoding,li2018multi}. However, the study of multi-label emotion decoding is still lacking. Besides, in multi-view learning, the information in some views are useful to handle the weakness of other views. Furthermore, multi-view learning can naturally be embedded into multi-label learning tasks to improve the classification performance. Models in multi-view multi-label learning can be divided into generative models and discriminative models. \cite{sun2020lcbm} proposed a Gaussian mixture VAE placing a conditional Bernoulli mixtures distribution on the labels $\textbf{y}$, which belongs to a probabilistic generative model. \cite{zhang2018latent} used matrix factorization to uncover the latent patterns among different views for more accurate multi-label classification. \cite{wu2019multi} exploited shared subspace for fusing multi-view representations and also took view-specific information into consideration. The above two models are both attributed to the discriminative model. Besides, \cite{kuleshov2017deep} proposed deep hybrid models which can bridge the discriminative and generative models in which the former often attain higher predictive accuracy, while the latter are more strongly regularized and can obtain more expressive representations. The discriminative component of the hybrid model pays more attention to the predict probability $p(\textbf{y}|\textbf{x})$ rather than the distribution of the label $p(\textbf{y})$. Inspired by this, we propose a multi-view multi-label hybrid model for learning the expressive neural representations and predicting multiple emotional states accurately.

\section{Methodology}
\subsection{Problem definition}
Formally speaking, we have a multi-view multi-label dataset with $N$ samples $\mathcal{D} = \left\lbrace (\mathbf{x}^{l}_{i},\mathbf{x}^{r}_{i},\mathbf{x}^{d}_{i},\mathbf{y}_i) |1\leq i \leq N \right\rbrace $, where $\mathbf{x}^{l}_{i}\in\mathbb{R}^{D_l}$ and $ \mathbf{x}^{r}_{i}\in \mathbb{R}^{D_r}$ are feature vectors of the left and right hemisphere views respectively, $\mathbf{x}^{d}_{i}$ is the difference view between the left and right hemispheres (i.e. $\mathbf{x}^{d}_{i}=\mathbf{x}^{l}_{i}-\mathbf{x}^{r}_{i}$), $\mathbf{y}_i \in \mathbb{R}^{C}$ is label vector in which $y_{ic}\in \left\lbrace 0,1 \right\rbrace $. $y_{ic}=1$ means that emotion label $c$ is relevant to the brain activity $i$. We introduce shared latent variables $\mathbf{z}$ which can be regarded as multi-view neural representations and are also capable of predicting multiple emotion labels. For simplicity, we rewrite the joint of three views $(\mathbf{x}^{l},\mathbf{x}^{r},\mathbf{x}^{d})$ as $\mathbf{x}$. The task of our hybrid model ML-BVAE is to learn a joint model $p(\mathbf{x},\mathbf{y},\mathbf{z})$ from $\mathcal{D}$. We assume that $p(\mathbf{x},\mathbf{y},\mathbf{z})$ has a parametric form specified by the decomposition:
\begin{equation}
	p(\mathbf{x},\mathbf{y},\mathbf{z}) = p(\mathbf{y}|\mathbf{x},\mathbf{z})p(\mathbf{x},\mathbf{z}),
\end{equation}
in which $p(\mathbf{x},\mathbf{z})$ and $p(\mathbf{y}|\mathbf{x},\mathbf{z})$ are, respectively, the latent variables generative component and discriminative component of ML-BVAE. In practice, we assume that the latent variables $\mathbf{z}$ contain all the information of $\mathbf{x}$. Therefore, we simplify the discriminative component as $p(\mathbf{y}|\mathbf{z})$.

The standard approach for training ML-BVAE is to maximize the marginal likelihood:
\begin{equation}
	\textup{ log }p(\mathbf{x},\mathbf{y})= \textup{ log}\int_{\mathbf{z}}p(\mathbf{x},\mathbf{y},\mathbf{z}) = \textup{ log}\int_{\mathbf{z}}p(\mathbf{y}|\mathbf{x},\mathbf{z})p(\mathbf{x},\mathbf{z}).
\end{equation}
For the intractable integral, we apply variational inference to obtain a tight lower bound of the marginal likelihood with Jensen's inequality :
\begin{align}
	\nonumber &	\textup{ log }\int_{\mathbf{z}}p(\mathbf{y}|\mathbf{x},\mathbf{z})p(\mathbf{x},\mathbf{z})\\
	\nonumber		&=\textup{ log }\int_{\mathbf{z}}p(\mathbf{y}|\mathbf{z})p(\mathbf{x},\mathbf{z})\\
	\nonumber &= \textup{ log }\int_{\mathbf{z}} \frac{p(\mathbf{y}|\mathbf{z})p(\mathbf{x},\mathbf{z})}{q(\mathbf{z}|\mathbf{x})}q(\mathbf{z}|\mathbf{x})\\
	\nonumber & \geq \mathbb{E}_{q(\mathbf{z}|\mathbf{x})}[\textup{ log }p(\mathbf{x},\mathbf{z})-\textup{ log }q(\mathbf{z}|\mathbf{x})+\textup{ log }p(\mathbf{y}|\mathbf{z})]\\
	& = \mathcal{L}_{ELBO} + \mathbb{E}_{q(\mathbf{z}|\mathbf{x})}[\textup{ log }p(\mathbf{y}|\mathbf{z})],
	\label{likelihood}
\end{align}
in which $q(\mathbf{z}|\mathbf{x})$ is the variational distribution. The lower bound of ML-BVAE has two terms: $\mathcal{L}_{ELBO}$ denotes the evidence lower bound (ELBO) of the generative component, while $\mathbb{E}_{q(\mathbf{z}|\mathbf{x})}[\textup{ log }p(\mathbf{y}|\mathbf{z})]$ can measure the classification error of the discriminative component (the larger the value, the smaller the classification error.).

After training, we can obtain the multi-label classification probability by:
\begin{equation}
	p(\mathbf{y}|\mathbf{x}) = p(\mathbf{y}|\mathbf{z})p(\mathbf{z}|\mathbf{x}) \approx p(\mathbf{y}|\mathbf{z})q(\mathbf{z}|\mathbf{x}).
\end{equation}
\begin{figure}
	\centering
	\includegraphics[height=1.2in,width=3.1in]{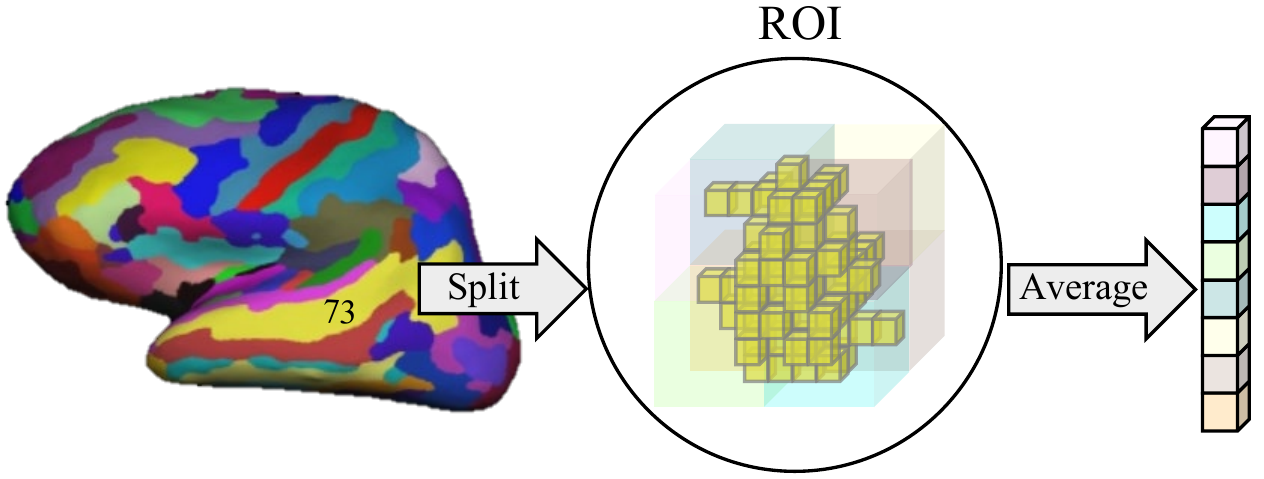}
	\caption{A schematic diagram of ROI pooling with Destrieux atlas \cite{destrieux2010automatic}. Take the 73th ROI \emph{Superior temporal sulcus} as an example, each yellow small cube in this ROI represents a voxel. Voxels in this ROI are placed in a 3-D volume which is spilted evenly into several subvolumes. Then the brain activity of voxels in each subvolume are averaged.}
	\label{fig:roi}
\end{figure}
\subsection{Overview}
The proposed approach involves three key components: \emph{fMRI Dimensionality Reduction with ROI Pooling}, \emph{Multi-view Neural Representations Learning with BVAE} and \emph{Multi-label Learning with Emotion-specific Neural Representations}, in which the second and the third components shown in Fig. \ref{fig:mlbvae} serve as the generative and the discriminative component of ML-BVAE, respectively. Below, we will introduce them separately.

\begin{figure*}[!hbt]
	\centering
	\includegraphics[height=2.8in,width=7.2in]{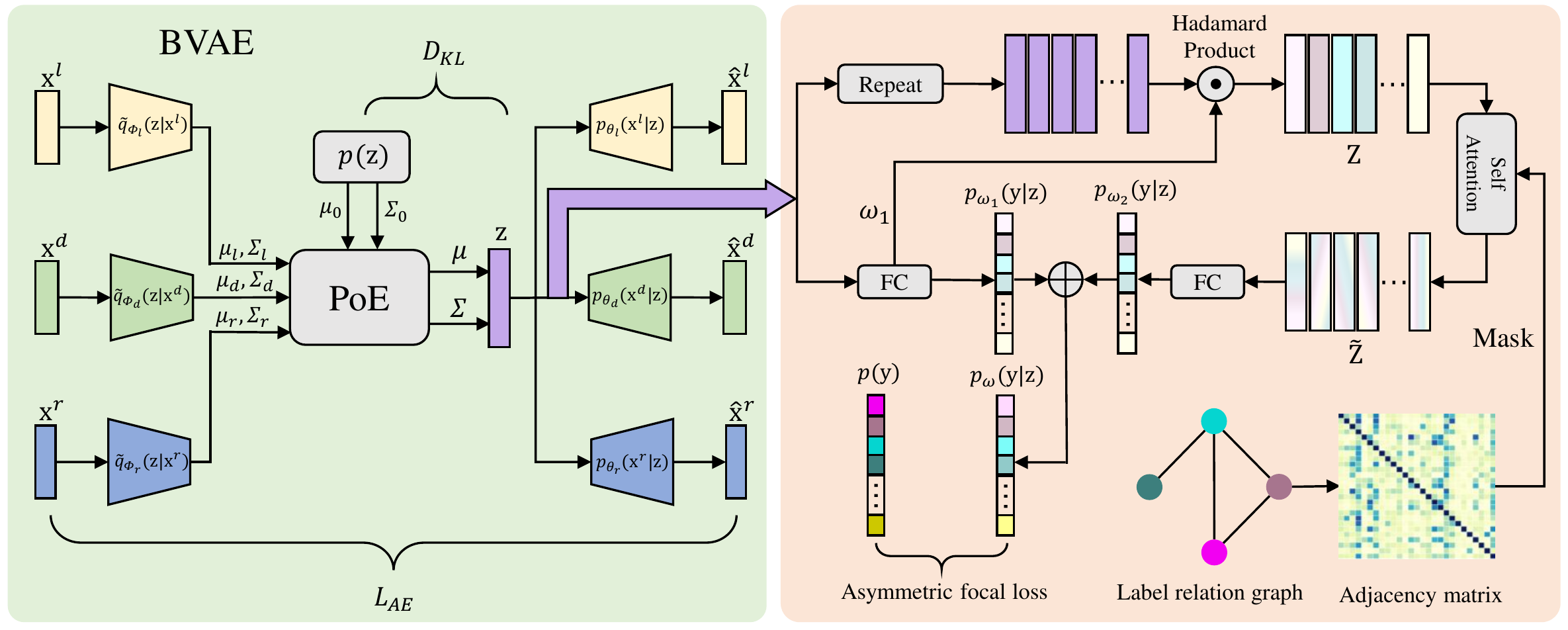}
	\caption{An overview of the proposed ML-BVAE approach. The green part is the BVAE which serves as the generative component and the orange part is the multi-label classification network with a label-aware module and a masked self-attention module which is the discriminative component. The label relation graph is constructed based on the label co-occurrence relationship.}
	\label{fig:mlbvae}
	\vspace{-0.2cm}
\end{figure*}
\subsection{fMRI Dimensionality Reduction with ROI Pooling}
Using fMRI signal directly for voxel based decoding will introduce considerable noise and easily cause overfitting due to high data dimensionality, small sample size and low signal-to-noise ratio \cite{beliy2019voxels,li2020visual,du2020conditional}. Therefore, we first use brain atlas \cite{glasser2016multi,destrieux2010automatic} to divide the whole brain cortical into multiple brain areas (ROIs). In order to further extract the features of each ROI, we place the voxels of each ROI in a 3-D volume according to its coordinates, then split the volume evenly into several subvolumes and calculate the average brain activity of voxels in each subvolume as the feature of this subvolume as illustrated in Fig. \ref{fig:roi} (take 8 sub-volumes for example). Then we concatenate the features of ROI in each hemisphere to obtain ROI pooling features $\mathbf{x}^l$ and $\mathbf{x}^r$. This method can realize noise suppression and dimensionality reduction of the input signal. We set the number of features in each ROI as an optional hyperparameter named $N_{ROIF}$. 

\subsection{Multi-view Neural Representations Learning with BVAE}

In BVAE, we assume that the left hemisphere view $\mathbf{x}^{l}$, the right hemisphere view $\mathbf{x}^{r}$ and the difference view $\mathbf{x}^{d}$ can be generated by the common latent variables $\mathbf{z}$, which can be regarded as the common features among the three distinct views. That is we assume a generative model of the form:
\begin{equation}
	\label{eq:generative model}
	p_{\theta}(\mathbf{x},\mathbf{z}) = p(\mathbf{z})p_{\theta_l}(\mathbf{x}^{l}|\mathbf{z})p_{\theta_r}(\mathbf{x}^{r}|\mathbf{z})p_{\theta_d}(\mathbf{x}^{d}|\mathbf{z}),
\end{equation}
where $\theta = \left\lbrace \theta_l,\theta_r,\theta_d\right\rbrace $. Then the three views ELBO, written as $\mathcal{L}_{ELBO_{lrd}}$, is equal to:
\begin{align}\label{Eq:mmelbo}
	\nonumber	&\underbrace{\mathbb{E}_{q_{\phi}}[ \lambda_{l} \textup{ log }p_{\theta_l}(\mathbf{x}^{l}|\mathbf{z})+\lambda_{r} \textup{ log }p_{\theta_r}(\mathbf{x}^{r}|\mathbf{z})+\lambda_{d} \textup{ log }p_{\theta_d}(\mathbf{x}^{d}|\mathbf{z})]}_{L_{AE}}\\
	&-\underbrace{\beta \textup{ KL}[q_{\phi}(\mathbf{z}|\mathbf{x})||p(\mathbf{z})]}_{D_{KL}},
\end{align}
where the $L_{AE}$ term denotes the reconstruction error of the three views and the $D_{KL}$ term regularizes the encoder by minimizing the Kullback-Leibler divergence between the joint approximate posterior $q_{\phi}(\mathbf{z}|\mathbf{x})$ and the prior $p(\mathbf{z})$.  In practice, $\beta$ is slowly annealed to 1 \cite{higgins2016beta,bowman2015generating} to form a valid ELBO. $\lambda_l$,  $\lambda_r$ and  $\lambda_d$ are hyperparameters that balance the reconstruction error among the three views.

The key issue for training BVAE is to specify the joint approximate posterior distribution  $q_{\phi}(\mathbf{z}|\mathbf{x})$. According to the conditional independence assumptions in the generative model, the true joint posterior distribution can be written as:
\begin{align}
	\label{eqn:derive}
	\nonumber p(\mathbf{z}|\mathbf{x}) & = \frac{p(\mathbf{x}|\mathbf{z})p(\mathbf{z})}{p(\mathbf{x})}=\frac{p(\mathbf{z})}{p(\mathbf{x})}p(\mathbf{x}^{l}|\mathbf{z})p(\mathbf{x}^{r}|\mathbf{z})p(\mathbf{x}^{d}|\mathbf{z})\\
	\nonumber &=\frac{p(\mathbf{z})}{p(\mathbf{x})}\frac{p(\mathbf{z}|\mathbf{x}^l)p(\mathbf{x}^l)}{p(\mathbf{z})}	\frac{p(\mathbf{z}|\mathbf{x}^r)p(\mathbf{x}^r)}{p(\mathbf{z})} \frac{p(\mathbf{z}|\mathbf{x}^d)p(\mathbf{x}^d)}{p(\mathbf{z})}\\
	\nonumber &=\frac{p(\mathbf{z}|\mathbf{x}^l)p(\mathbf{z}|\mathbf{x}^r)p(\mathbf{z}|\mathbf{x}^d)}{p^2(\mathbf{z})}\frac{p(\mathbf{x}^l)p(\mathbf{x}^r)p(\mathbf{x}^d)}{p(\mathbf{x})}\\
	&\propto\frac{p(\mathbf{z}|\mathbf{x}^l)p(\mathbf{z}|\mathbf{x}^r)p(\mathbf{z}|\mathbf{x}^d)}{p^2(\mathbf{z})}.
\end{align}
which means the joint posterior is a product of individual posterior from each view, with an additional quotient by the square of prior.
For numerical stability, if we approximate the single view true posterior with $q_{\phi_l}(\mathbf{z}|\mathbf{x}^l) = \tilde{q}_{\phi_l}(\mathbf{z}|\mathbf{x}^l)p(\mathbf{z})$, $q_{\phi_r}(\mathbf{z}|\mathbf{x}^r) = \tilde{q}_{\phi_r}(\mathbf{z}|\mathbf{x}^r)p(\mathbf{z})$ and $q_{\phi_d}(\mathbf{z}|\mathbf{x}^d) = \tilde{q}_{\phi_d}(\mathbf{z}|\mathbf{x}^d)p(\mathbf{z})$, where $\tilde{q}_{\phi_l}(\mathbf{z}|\mathbf{x}^l)$, $\tilde{q}_{\phi_r}(\mathbf{z}|\mathbf{x}^r)$ and $\tilde{q}_{\phi_d}(\mathbf{z}|\mathbf{x}^d)$ are the underlying inference network of each view, we can avoid the quotient term by:
\begin{align}
	\label{eqn:derive:simple}
	\nonumber	p(\mathbf{z}|\mathbf{x})
	\nonumber	&\propto \frac{p(\mathbf{z}|\mathbf{x}^l)p(\mathbf{z}|\mathbf{x}^r)p(\mathbf{z}|\mathbf{x}^d)}{p^2(\mathbf{z})}\\
	\nonumber	&\approx \frac{[\tilde{q}_{\phi_l}(\mathbf{z}|\mathbf{x}^l)p(\mathbf{z})][\tilde{q}_{\phi_r}(\mathbf{z}|\mathbf{x}^r)p(\mathbf{z})][\tilde{q}_{\phi_d}(\mathbf{z}|\mathbf{x}^d)p(\mathbf{z})]}{p^2(\mathbf{z})}\\
	&= p(\mathbf{z})\tilde{q}_{\phi_l}(\mathbf{z}|\mathbf{x}^l)\tilde{q}_{\phi_r}(\mathbf{z}|\mathbf{x}^r)\tilde{q}_{\phi_d}(\mathbf{z}|\mathbf{x}^d). 
\end{align}
That is, the encoder of BVAE can be written as:
\begin{align}
	q_\phi(\mathbf{z}|\mathbf{x}) = p(\mathbf{z})\tilde{q}_{\phi_l}(\mathbf{z}|\mathbf{x}^l)\tilde{q}_{\phi_r}(\mathbf{z}|\mathbf{x}^r)\tilde{q}_{\phi_d}(\mathbf{z}|\mathbf{x}^d), 
\end{align}
in which $\phi = \left\lbrace \phi_l,\phi_r,\phi_d\right\rbrace $. In other words, the joint approximate posterior distribution is a product of prior and three single view posterior distributions called product of experts (PoE) \cite{wu2018multimodal}. When these four distributions are all Gaussian which means $p(\mathbf{z})=\mathcal{N}(\boldsymbol{\mu}_0|\mathbf{\Sigma}_0)$,  $\tilde{q}_{\phi_l}(\mathbf{z}|\mathbf{x}^l)=\mathcal{N}(\boldsymbol{\mu}_l|\mathbf{\Sigma}_l)$, $\tilde{q}_{\phi_r}(\mathbf{z}|\mathbf{x}^r)=\mathcal{N}(\boldsymbol{\mu}_r|\mathbf{\Sigma}_r)$ and $\tilde{q}_{\phi_d}(\mathbf{z}|\mathbf{x}^d)=\mathcal{N}(\boldsymbol{\mu}_d|\mathbf{\Sigma}_d)$, there is an analytical solution of the product distribution acquired above \cite{cao2014generalized}:

\begin{equation}
	q_\phi(\mathbf{z}|\mathbf{x})=\mathcal{N}(\boldsymbol{\mu}|\mathbf{\Sigma}),
\end{equation}
which means $q_\phi(\mathbf{z}|\mathbf{x})$ is Gaussian with:
\begin{align}
	\nonumber \boldsymbol{\mu} &= (\mathbf{\Sigma}_0^{-1}\boldsymbol{\mu}_0+\mathbf{\Sigma}_l^{-1}\boldsymbol{\mu}_l+\mathbf{\Sigma}_r^{-1}\boldsymbol{\mu}_r+ \mathbf{\Sigma}_d^{-1}\boldsymbol{\mu}_d)\mathbf{\Sigma},\\
	\mathbf{\Sigma} &= (\mathbf{\Sigma}_0^{-1}+\mathbf{\Sigma}_l^{-1}+\mathbf{\Sigma}_r^{-1}+\mathbf{\Sigma}_d^{-1})^{-1},
\end{align}
whose proof can be found in Appendix A. In practice, we assume that $p(\mathbf{z})$ is spherical Gaussian prior with $\boldsymbol{\mu}_0 = \mathbf{0}$ and $\mathbf{\Sigma}_0 = \mathbf{I}$.

BVAE supports single view input and generates three views at the same time due to the product property of its encoder. In order to model the information intra- and inter- the three views, in addition to the ELBO defined by Eq. (\ref{Eq:mmelbo}), we define another three ELBOs.
\begin{align}
	\label{Eq:mmelbo1}
	\nonumber \mathcal{L}_{ELBO_l} &= \underbrace{\mathbb{E}_{q_{\phi_l}}[ \lambda_{l} \textup{ log }p_{\theta_l}(\mathbf{x}^{l}|\mathbf{z})]}_{L_{AE_{ll}}} + \underbrace{\mathbb{E}_{q_{\phi_l}}[\lambda_{r} \textup{ log }p_{\theta_r}(\mathbf{x}^{r}|\mathbf{z})]}_{L_{AE_{lr}}} \\
	&+ \underbrace{\mathbb{E}_{q_{\phi_l}}[\lambda_{d} \textup{ log }p_{\theta_r}(\mathbf{x}^{d}|\mathbf{z})]}_{L_{AE_{ld}}} - \underbrace{\beta \textup{ KL}[q_{\phi_l}(\mathbf{z}|\mathbf{x}^{l})|| p(\mathbf{z})]}_{D_{KL_{l}}},
\end{align}
\begin{align}
	\label{Eq:mmelbo2}
	\nonumber \mathcal{L}_{ELBO_r} &= \underbrace{\mathbb{E}_{q_{\phi_r}}[ \lambda_{l} \textup{ log }p_{\theta_l}(\mathbf{x}^{l}|\mathbf{z})]}_{L_{AE_{rl}}} + \underbrace{\mathbb{E}_{q_{\phi_r}}[\lambda_{r} \textup{ log }p_{\theta_r}(\mathbf{x}^{r}|\mathbf{z})]}_{L_{AE_{rr}}} \\
	&+ \underbrace{\mathbb{E}_{q_{\phi_r}}[\lambda_{d} \textup{ log }p_{\theta_r}(\mathbf{x}^{d}|\mathbf{z})]}_{L_{AE_{rd}}} - \underbrace{\beta \textup{ KL}[q_{\phi_r}(\mathbf{z}|\mathbf{x}^{r})|| p(\mathbf{z})]}_{D_{KL_{r}}},
\end{align}
\begin{align}
	\label{Eq:mmelbo3}
	\nonumber \mathcal{L}_{ELBO_d} &= \underbrace{\mathbb{E}_{q_{\phi_d}}[ \lambda_{l} \textup{ log }p_{\theta_l}(\mathbf{x}^{l}|\mathbf{z})]}_{L_{AE_{dl}}} + \underbrace{\mathbb{E}_{q_{\phi_d}}[\lambda_{r} \textup{ log }p_{\theta_r}(\mathbf{x}^{r}|\mathbf{z})]}_{L_{AE_{dr}}} \\
	&+ \underbrace{\mathbb{E}_{q_{\phi_d}}[\lambda_{d} \textup{ log }p_{\theta_r}(\mathbf{x}^{d}|\mathbf{z})]}_{L_{AE_{dd}}} - \underbrace{\beta \textup{ KL}[q_{\phi_d}(\mathbf{z}|\mathbf{x}^{d})|| p(\mathbf{z})]}_{D_{KL_{d}}},
\end{align}
where $\mathcal{L}_{ELBO_l}$, $\mathcal{L}_{ELBO_r}$ and $\mathcal{L}_{ELBO_d}$ mean only using the left, right and difference single view as the input to the BVAE, respectively. The $L_{AE}$ terms with the same subscript are reconstruction errors intra-view and the $L_{AE}$ terms with different subscript are reconstruction errors inter-view. $D_{KL_{l}}$, $D_{KL_{r}}$ and $D_{KL_{d}}$ are KL divergence between the three views encoder and the prior, respectively. Above all, the total ELBO of BVAE can be written as (Actually, $\mathcal{L}_{ELBO_{lrd}}$ is the lower bound on the joint loglikelihood and another three ELBOs are useful for training individual encoders \cite{wu2018multimodal}.) :
\begin{align}
	\mathcal{L}_{ELBO} &= \mathcal{L}_{ELBO_{lrd}}+\mathcal{L}_{ELBO_l} + \mathcal{L}_{ELBO_r} + \mathcal{L}_{ELBO_d}.
	\label{total_elbo}
\end{align}

\subsection{Multi-label Learning with Emotion-specific Neural Representations}
In order to further model the relationship between multi-view neural representations and emotion labels, we use a label-aware module \cite{chen2021learning} to get a set of emotion-specific neural representations from $\mathbf{z}$ by utilizing the distinct discriminative properties of each label. Firstly, we apply classifiers which can be implemented by fully connected layers to the multi-view neural representations $\mathbf{z}\in \mathbb{R}^D$:
\begin{equation}
	p_{\omega_{1}}({\mathbf{y}}|{\mathbf{z}}) = \textup{sigmoid}[f(\mathbf{z};\omega_{1})],
\end{equation}
where $\omega_{1} \in \mathbb{R}^{C\times D}$ represents the parameters of these classifiers and $C$ is the number of emotion labels. Each classifier $\omega_{1}^{c}\in \mathbb{R}^D$ extracts the information corresponding to emotion class $c$ and predicts the probability of emotion class $c$ from the brain activity. Then the emotion-specific neural representations can be obtained via:
\begin{equation}
	\mathbf{Z} = \textup{repeat}(\mathbf{z})\odot \omega_{1} \in \mathbb{R}^{C \times D},
\end{equation}
in which $\textup{repeat}(\mathbf{z})$ indicates the operation of copying $\mathbf{z}\in \mathbb{R}^D$ for $C$ times and $\odot$ denotes the Hadamard product. In this way, $\mathbf{Z}^c \in \mathbb{R}^D$ captures the information related to the emotion class $c$. 

To further model the dependency of emotion labels, we take use of a masked self-attention operation \cite{vaswani2017attention} across the emotion-specific neural representations. In this module, emotion-specific neural representations $\mathbf{Z}$ generate a set of query, key and value matrices $(\mathbf{Q},\mathbf{K},\mathbf{V})$ by three linear transformations whose parameters are $\left\lbrace \omega_q, \omega_k,\omega_v \right\rbrace $
Then a $C\times C$ attention matrix $\mathbf{A}$ is obtained as follows:
\begin{equation}
	\mathbf{A} = \textup{softmax}(\frac{\mathbf{Q}\mathbf{K}^T}{\sqrt{d'}}+\textup{Mask}) ,
	\label{attention_mat}
\end{equation}
\begin{equation}
	\textup{Mask}(j,k) = \frac{N_{j,k}}{N_{j}},
\end{equation}
where $d'$ is the dim of query and key. $\textup{Mask}\in \mathbb{R}^{C\times C}$ is the label graph adjacency matrix \cite{zhang2021multi} in which $\textup{Mask}(j,k)$ is the frequency of emotion $j$ co-occurrence with emotion $k$ in the training ground truth, $N_{j,k}$ is the number of training samples that include both emotion $j$ and $k$, and $N_j$ is the number of training samples in which emotion $j$ occurs. By this, we apply an emotion label-wise correlation map $\mathbf{A}$ onto the emotion specific neural representations. Then we can get the refined emotion-specific neural representations $\tilde{\mathbf{Z}}$ by:
\begin{equation}
	\tilde{\mathbf{Z}} = \mathbf{A}\mathbf{V}.
\end{equation}
Then we use another fully connected layers to the refined emotion-specific neural representations:
\begin{equation}
	p_{\omega_{2}}({\mathbf{y}}|{\mathbf{z}}) = \textup{sigmoid}[g(\tilde{\mathbf{Z}};\omega_{g})],
\end{equation}
where $\omega_{2} = \left\lbrace\omega_q, \omega_k,\omega_v, \omega_{g}\right\rbrace $. The final output of the ML-BVAE is:
\begin{equation}
	p_{\omega}({\mathbf{y}}|{\mathbf{z}})=p_{\omega_{1}}({\mathbf{y}}|{\mathbf{z}})+p_{\omega_{2}}({\mathbf{y}}|{\mathbf{z}}),
\end{equation}
in which $\omega = \left\lbrace \omega_{1},\omega_{2}\right\rbrace $. 
\subsection{Loss function}
In supervised scenario, to more effectively address the sample imbalance problem, we substitute the log-likelihood in the classification error term with asymmetric focal loss \cite{ben2020asymmetric,liu2021query2label}:
\begin{equation}
	\begin{aligned}
		\textup{ log } p_{\omega}(\mathbf{y}|\mathbf{z}) \approx \frac{1}{C}
		\sum_{c=1}^{C}
		\begin{cases}
			(1-p_\omega^c)^{\gamma^+}\log (p_\omega^c), & y_c=1,\\
			(p_\omega^c)^{\gamma^-}\log (1-p_\omega^c), & y_c=0,\\
		\end{cases}
	\end{aligned}
	\label{eq:loss}
\end{equation}
where 
$y_c$ is the label that indicates whether brain activity $\mathbf{x}$ has emotional state $c$ and $p_\omega^c$ is the $c$-th component of $p_\omega(\mathbf{y}|\mathbf{z})$. $\gamma^+$ and $\gamma^-$ are hyperparameters that balance the loss of positive and negative labels. Above all, the lower bound of the marginal likelihood can be:
\begin{equation}
	\mathcal{L}(\phi,\theta,\omega) = \lambda_{1}\mathcal{L}_{ELBO}(\phi,\theta) + \lambda_{2} \mathbb{E}_{q_{\phi}(\mathbf{z}|\mathbf{x})}[\textup{ log }p_{\omega}(\mathbf{y}|\mathbf{z})],
	\label{loss_function}
\end{equation}
where $\lambda_{1}$ and $\lambda_{2}$ are hyper parameters that balance the loss between the generative and the discriminative component. After training, given a test brain activity pattern $\mathbf{x}$, we can predict the probability of multiple emotional states by:
\begin{equation}
	p(\mathbf{y}|\mathbf{x}) = p_{\omega^*}(\mathbf{y}|\mathbf{z})q_{\phi^*}(\mathbf{z}|\mathbf{x}),
\end{equation}
where $\omega^*$ and $\phi^*$ are optimal parameters of the multi-label classification network and the encoder of BVAE, respectively.
\begin{table}[H]
	\centering
	\caption{The details of the datasets used in our experiments.}
	\label{data sets}
	\resizebox{\columnwidth}{!}{	
		\begin{tabular}{c|c c  c c c}
			\hline
			\hline
			Dataset  & $\#$Subjects        & $\#$Instances              &$\#$Voxels       &$\#$ROIs      &$\#$Emotions           \\ \hline
			MEMO27        & 5                   & 2196                       & 115070       &360             &27          \\
			MEMO80       & 8                   & 5400                       & 83693        &148             &80        \\
			\hline
			\hline
	\end{tabular}}
\end{table}

\section{Experiments} \label{experiments}

\subsection{Datasets}
The properties of the two visually evoked emotional datasets used in our experiment are summarized in Table \ref{data sets}. Below, we will introduce them in detail.
\subsubsection{MEMO27}
MEMO27 is a fine-grained \textbf{M}ulti-label \textbf{EMO}tion decoding dataset based on a publicly available fMRI dataset\footnote{available at https://doi.org/10.6084/m9.figshare.11988351.v1}, which contains the blood-oxygen-level dependent (BOLD) responses of five subjects. fMRI data were collected using a 3T Siemens scanner with a multiband gradient Echo-Planar Imaging (EPI) sequence (TR, 2000 ms; TE, 43 ms; flip angle, 80 deg; FOV, 192 × 192 mm; voxel size, 2 × 2 × 2 mm; number of slices, 76; multiband factor, 4) \cite{horikawa2020neural}. Subjects were presented with 2196 videos whose durations ranged from 0.15s to 90s \cite{horikawa2020neural}. These videos can provoke a variety of emotions. Each video was voted by multiple raters across 34 emotion categories and the voting ratio of each emotion category was used as the rating of that emotion category. The fMRI data was preprocessed and averaged with each video stimulus \cite{horikawa2020neural} which means the brain activity of one voxel is a scalar for a video stimulus. We use HCP360 atlas defined in \cite{glasser2016multi} to segment the cerebral cortex of each subject in MEMO27 into 360 ROIs (180 ROIs per hemisphere) as illustrated in \cite{horikawa2020neural}.

As for the emotion ratings, we first choose 27 of 34 emotion categories according to \cite{cowen2017self}. Then we set a threshold of 0.1 to construct an emotion label matrix from emotion ratings. The average number of emotion labels is 4.64 and the corresponding label density is 0.172. 
\subsubsection{MEMO80}
MEMO80 is a more fine-grained \textbf{M}ulti-label \textbf{EMO}tion decoding dataset\footnote{Restricted by copyright, MEMO80 is available once the corresponding author of \cite{koide2020distinct} gives permission.}, which contains the BOLD responses of eight subjects. fMRI data were collected using a 3T Siemens scanner with a multiband gradient EPI sequence (TR, 2000 ms; TE, 30 ms; flip angle, 62 deg; FOV, 192 × 192 mm; voxel size, 2 × 2 × 2 mm; number of slices, 72; multiband factor, 3) \cite{koide2020distinct}. Subjects were presented with emotion-inducing audiovisual movies over a period of three hours. For each of the one-second movie scenes, the movie was annotated with respect to 80 emotions by multiple raters. The fMRI data and emotion ratings were preprocessed as illustrated in \cite{koide2020distinct}. According to \cite{koide2020distinct}, each sampling moment of fMRI can be regarded as a sample which leads to 5400 samples in total for one subject. We use Destrieux atlas defined in \cite{destrieux2010automatic} to segment the cerebral cortex of each subject in MEMO80 into 148 ROIs (74 ROIs per hemisphere) as illustrated in \cite{koide2020distinct}.

As for the emotion ratings, we exploit all the 80 emotion categories for the completeness and diversity. Then we set a threshold of 0.5 to construct an emotion label matrix from emotion ratings. The average number of emotion labels is 17.9 and the corresponding label density is 0.224.

\subsection{Compared methods}
The performance of ML-BVAE is compared against six algorithms, including a linear model, three single view multi-label learning methods and two multi-view multi-label learning methods. (1) \textbf{linear regression (LR)} \cite{horikawa2020neural}: A traditional fMRI decoding method. In order to be consistent with the settings of \cite{horikawa2020neural}, we trained a linear model for each emotion category with all cortical voxels. Then we binarized the predicted ratings using the same threshold as the label matrix construction. (2) \textbf{Benchmark}: A multi-layer perceptron (MLP) where the number of parameters is no less than ML-BVAE. The input is the concatenation of all view features; (3) \textbf{CA2E} \cite{yeh2017learning}: The first deep neural networks (DNNs) based multi-label learning model which can learn deep latent space by mapping the feature and label jointly. The hyperparameter $\alpha$ is set to 2.0 and the dimension $l$ of the latent space is set to 128 and 64 for MEMO27 and MEMO80, respectively. (4) \textbf{ML-GCN} \cite{chen2019multi}: A multi-label learning method based on GCNs. We adopt MLP used in Benchmark for feature extraction; (5) \textbf{SIMM} \cite{wu2019multi}: A multi-view multi-label learning method which leverages shared subspace exploitation and view-specific information extraction. The hyper parameters $\alpha$ and $\beta$ are both set to 0.1; (6) \textbf{GARDIS} \cite{chen2020multi}: A multi-view multi-label learning method based on graph and spectral clustering. The hyperparameters $ \left\lbrace  \alpha, \gamma,\eta,k\right\rbrace $  are set to $\left\lbrace 0.95, 0.1, 0.1,8\right\rbrace $, respectively.
\newcolumntype{b}{>{\columncolor{mygray}}c}
\begin{table}
	\centering
	\caption{Multi-label emotion decoding performance of each compared approach in terms of One-Error, Ranking Loss, Micro F1, Macro F1, example-based Average Precision and mean Average Precision in MEMO27. $\uparrow$ ($\downarrow$) indicates the larger (smaller) the value, the better the performance.}
	\label{of1}
	\scriptsize
	\resizebox{\columnwidth}{!}{
		\begin{tabular}{c|c|cccccc} 
			\hline
			\hline
			Subject & Method & OneE$\downarrow$ & RL$\downarrow$ & miF1$\uparrow$& maF1$\uparrow$&e-AP$\uparrow$&mAP$\uparrow$\\
			\hline
			& LR\cite{horikawa2020neural} & 0.546 & 0.506 & 0.197& 0.155&0.309&0.265\\
			& Benchmark & 0.315 & 0.190 & 0.449& 0.336&0.609&0.438\\
			&CA2E\cite{yeh2017learning}&0.442&0.255&0.428&0.345&0.519&0.354\\
			&  ML-GCN\cite{chen2019multi}  & 0.307  & 0.190  & 0.445  & 0.331  & 0.606&0.421  \\
			&  SIMM\cite{wu2019multi}  & 0.360  & 0.217  & 0.396  &0.239   &0.573 &0.399  \\
			&  GARDIS\cite{chen2020multi}  &0.315   & 0.230  & 0.228    & 0.120  &0.580 & 0.394\\
			\cline{2-8}
			\multirow{-7}{*}{Subject1}
			&ML-BVAE  & \textbf{0.290}  & \textbf{0.183}  & \textbf{0.504}  & \textbf{0.398} & \textbf{0.620}& \textbf{0.451} \\
			\hline
			\hline
			\multirow{7}{*}{Subject2}
			& LR\cite{horikawa2020neural} & 0.525 & 0.493 & 0.192& 0.141&0.317&0.251\\
			& Benchmark & 0.320 & 0.182 &0.457 & 0.330&0.619&0.447\\
			&CA2E\cite{yeh2017learning}&0.427&0.249&0.427&0.326&0.527&0.353\\
			&  ML-GCN\cite{chen2019multi}  & 0.309  & 0.181  & 0.450  & 0.324  & 0.621 & 0.441\\
			&  SIMM\cite{wu2019multi}  & 0.363  & 0.205  & 0.381  & 0.230  & 0.578 & 0.414\\
			&  GARDIS\cite{chen2020multi}  & 0.327  & 0.221  & 0.208    & 0.110  & 0.572&0.391 \\
			\cline{2-8}
			&ML-BVAE  & \textbf{0.288}  & \textbf{0.177}& \textbf{0.515}  & \textbf{0.410} & \textbf{0.632}& \textbf{0.466} \\
			\hline
			\hline
			\multirow{7}{*}{Subject3}
			& LR\cite{horikawa2020neural} & 0.628 & 0.504 & 0.157& 0.104&0.303&0.228\\
			& Benchmark & 0.308 & 0.182 & 0.452& 0.325&0.619&0.449\\
			&CA2E\cite{yeh2017learning}&0.459&0.259&0.413&0.318&0.503&0.337\\
			&  ML-GCN\cite{chen2019multi}  & 0.318  & 0.183  & 0.450  & 0.320  & 0.617&0.435  \\
			&  SIMM\cite{wu2019multi}  & 0.376  & 0.208  & 0.358  & 0.216  & 0.568&0.395  \\
			&  GARDIS\cite{chen2020multi}  & 0.349  & 0.224  & 0.199    & 0.096  & 0.563&0.377 \\
			\cline{2-8}
			&ML-BVAE  & \textbf{0.294}  & \textbf{0.181}  & \textbf{0.517}  & \textbf{0.414} & \textbf{0.627}&\textbf{0.461}  \\
			\hline
			\hline
			\multirow{7}{*}{Subject4}
			& LR\cite{horikawa2020neural} &0.525 & 0.488 &0.212& 0.144&0.330&0.259\\
			& Benchmark & 0.315 & 0.191 & 0.444& 0.319&0.608&0.429\\
			&CA2E\cite{yeh2017learning}&0.454&0.256&0.415&0.319&0.511&0.343\\
			&  ML-GCN\cite{chen2019multi}  & 0.321  & 0.191  & 0.444  & 0.317  & 0.603&0.421  \\
			&  SIMM\cite{wu2019multi}  & 0.364  & 0.207  & 0.378  & 0.235  & 0.576 &0.406 \\
			&  GARDIS\cite{chen2020multi}  & 0.330  & 0.222  & 0.231    & 0.125  & 0.574&0.378 \\
			\cline{2-8}
			&ML-BVAE  & \textbf{0.308}  & \textbf{0.186}  & \textbf{0.505}  & \textbf{0.391} & \textbf{0.618}&\textbf{0.441}  \\
			\hline
			\hline
			\multirow{7}{*}{Subject5}
			& LR\cite{horikawa2020neural} &0.658 & 0.526 & 0.128& 0.095&0.299&0.227\\
			& Benchmark & 0.339 & 0.212 & 0.415& 0.295&0.584&0.407\\
			&CA2E\cite{yeh2017learning}&0.479&0.272&0.405&0.308&0.494&0.319\\
			&  ML-GCN\cite{chen2019multi}  & 0.352  & 0.211  & 0.407  & 0.291  & 0.583&0.390  \\
			&  SIMM\cite{wu2019multi}  & 0.398  & 0.226  & 0.344  & 0.188  & 0.551&0.378  \\
			&  GARDIS\cite{chen2020multi}  & 0.368  & 0.244  & 0.172    & 0.077  & 0.548&0.355 \\
			\cline{2-8}
			&ML-BVAE  & \textbf{0.329}  & \textbf{0.204}  & \textbf{0.482}  & \textbf{0.375} & \textbf{0.596}&\textbf{0.423}  \\
			\hline
			\hline
			
			\multirow{7}{*}{Average}
			& LR\cite{horikawa2020neural} & 0.576 & 0.503 & 0.177& 0.128&0.312&0.246\\
			& Benchmark & 0.319 &0.191  &0.443 &0.321 &0.608 &0.434\\
			&CA2E\cite{yeh2017learning}& 0.452 &0.258 &0.418 &0.323 & 0.511&0.341\\
			&  ML-GCN\cite{chen2019multi}  & 0.321  &  0.191 & 0.439  & 0.317  &0.606 &  0.422\\
			&  SIMM\cite{wu2019multi}  & 0.372 & 0.213 &0.371  & 0.222  &0.569 &0.398  \\
			&  GARDIS\cite{chen2020multi}  & 0.338  & 0.228  &  0.208   & 0.106  & 0.567& 0.379\\
			\cline{2-8}
			&ML-BVAE  & \textbf{0.302}  & \textbf{0.186}  &\textbf{0.505}   &  \textbf{0.398} & \textbf{0.619} &  \textbf{0.448}  \\
			\hline
			\hline
	\end{tabular}}
	\label{MEMO27_all}
\end{table}
\subsection{Hyperparameters setting} For the $\mathcal{L}_{ELBO}$ in Eq. (\ref{total_elbo}), $\left\lbrace \lambda_{l} , \lambda_{r},\lambda_{d} \right\rbrace $ are all set to 1 and $\beta$ is slowly annealed to 1 according to:
\begin{equation}
	\beta(i,j) = \frac{j+(i-1) M+1}{\alpha M},
\end{equation}
where $i$ and $j$ are the epoch and batch index, respectively, and $M$ is the batch size which is set to 100.  $\alpha$ is the annealing ratio which is set to 100.
$\left\lbrace \gamma^+,\gamma^- \right\rbrace $ in asymmetric focal loss are set to 0 and 1 respectively  according to original reference. $ \lambda_{1}/\lambda_{2}$ in Eq. (\ref{loss_function}) is searched in $\left\lbrace 0.001,0.01,0.1,1,2,5,10,15,20 \right\rbrace $. $N_{ROIF}$ is searched in $\left\lbrace 8,27,64,125,216,343\right\rbrace $.  The dimensionality of neural representations $\mathbf{z}$ is set to 64 in both datasets. We train the model  using the Adam optimizer, with weight decay of 0.1, $\left\lbrace \beta_1, \beta_2\right\rbrace =\left\lbrace 0.9, 0.9999\right\rbrace $, and a learning rate of $ 10^{-4}$. 
\begin{table}[H]
	\centering
	\caption{Multi-label emotion decoding performance of each compared approach in terms of One-Error, Ranking Loss, Micro F1, Macro F1, example-based Average Precision and mean Average Precision in MEMO80. $\uparrow$ ($\downarrow$) indicates the larger (smaller) the value, the better the performance.}
	\label{of1}
	\scriptsize
	\resizebox{\columnwidth}{!}{
		\begin{tabular}{c|c|cccccc} 
			\hline
			\hline
			Subject & Method & OneE$\downarrow$ & RL$\downarrow$ & miF1$\uparrow$& maF1$\uparrow$&e-AP$\uparrow$&mAP$\uparrow$\\
			\hline
			\multirow{7}{*}{Subject1}
			& LR\cite{horikawa2020neural} & 0.536 & 0.439 & 0.289& 0.281&0.362&0.359\\
			& Benchmark & 0.446 & 0.316 & 0.227& 0.205&0.462&0.427\\
			&CA2E\cite{yeh2017learning}&0.564&0.389&0.358&0.348&0.394&0.344\\
			&  ML-GCN\cite{chen2019multi} &0.460 & 0.303  & 0.226  &0.206   &  0.472 & 0.455 \\
			&  SIMM\cite{wu2019multi}  &  0.482 & 0.316  & 0.263  &0.242   &0.455 &0.454  \\
			&  GARDIS\cite{chen2020multi}& 0.490 & 0.332  & 0.133  & 0.134    &  0.437 &0.419  \\
			\cline{2-8}
			&ML-BVAE  & \textbf{0.416}  & \textbf{0.295}  & \textbf{0.384}  & \textbf{0.374} & \textbf{0.488}& \textbf{0.474} \\
			\hline
			\hline
			\multirow{7}{*}{Subject2}
			& LR\cite{horikawa2020neural} & 0.488 & 0.424 & 0.345& 0.333&0.392&0.379\\
			& Benchmark & 0.444 & 0.301 &0.264 & 0.245&0.479&0.457\\
			&CA2E\cite{yeh2017learning}&0.543&0.378&0.358&0.345&0.404&0.358\\
			&  ML-GCN\cite{chen2019multi}  & 0.440  & 0.294  & 0.259  &0.244   &0.487  &0.466 \\
			&  SIMM\cite{wu2019multi}  &  0.474 & 0.311  & 0.285  &0.265   &0.467  &0.459\\
			&  GARDIS\cite{chen2020multi}  & 0.489  & 0.329  & 0.168    &0.166   &0.489 &0.425 \\
			\cline{2-8}
			&ML-BVAE  & \textbf{0.398}  & \textbf{0.284} & \textbf{0.396}  & \textbf{0.388} & \textbf{0.503}& \textbf{0.490} \\
			\hline
			\hline
			\multirow{7}{*}{Subject3}
			& LR\cite{horikawa2020neural} & 0.521 & 0.432 & 0.310& 0.303&0.373&0.364\\
			& Benchmark & 0.446 & 0.298 & 0.252& 0.229&0.474&0.455\\
			&CA2E\cite{yeh2017learning}&0.623&0.403&0.363&0.348&0.379&0.324\\
			&  ML-GCN\cite{chen2019multi}  & 0.453 &0.292   &0.248   &0.232   &0.484 &0.467  \\
			&  SIMM\cite{wu2019multi}  &  0.468 & 0.307  & 0.264  &0.246   &0.465 & 0.453 \\
			&  GARDIS\cite{chen2020multi}  & 0.507  &0.334  & 0.162   &0.161   & 0.436 & 0.414\\
			\cline{2-8}
			&ML-BVAE  & \textbf{0.417}  & \textbf{0.287}  & \textbf{0.394}  & \textbf{0.386} & \textbf{0.499}&\textbf{0.497}  \\
			\hline
			\hline
			\multirow{7}{*}{Subject4}
			& LR\cite{horikawa2020neural} & 0.531 & 0.439 & 0.298& 0.284&0.366&0.346\\
			& Benchmark & 0.470 & 0.336 & 0.179& 0.161&0.446&0.395\\
			&CA2E\cite{yeh2017learning}&0.549&0.391&0.347&0.329&0.392&0.331\\
			&  ML-GCN\cite{chen2019multi}  &  0.472 & 0.326  &0.195   &0.180   &0.460 & 0.424 \\
			&  SIMM\cite{wu2019multi}  & 0.468 & 0.328  & 0.235  & 0.217  &0.453 &0.419  \\
			&  GARDIS\cite{chen2020multi}  & 0.532  & 0.356  & 0.145    & 0.146  & 0.419 & 0.387\\
			\cline{2-8}
			&ML-BVAE  & \textbf{0.425}  & \textbf{0.311}  & \textbf{0.362}  & \textbf{0.353} & \textbf{0.478}&\textbf{0.455}  \\
			\hline
			\hline
			\multirow{7}{*}{Subject5}
			& LR\cite{horikawa2020neural} & 0.512 & 0.438 & 0.298& 0.286&0.366&0.359\\
			& Benchmark & 0.476 & 0.320 & 0.217& 0.200&0.458&0.423\\
			&CA2E\cite{yeh2017learning}&0.575&0.396&0.338&0.326&0.389&0.341\\
			&  ML-GCN\cite{chen2019multi}  &  0.439 & \textbf{0.306}   & 0.215  &0.202  &0.478 & 0.450 \\
			&  SIMM\cite{wu2019multi}  &  0.473 & 0.319  & 0.260  & 0.244  & 0.460 &0.443\\
			&  GARDIS\cite{chen2020multi}  & 0.505  & 0.341  & 0.146    &0.146   & 0.431& 0.411\\
			\cline{2-8}
			&ML-BVAE  & \textbf{0.416}  & \textbf{0.306}  & \textbf{0.370}  & \textbf{0.362} & \textbf{0.485}&\textbf{0.463}  \\
			\hline
			\hline
			\multirow{7}{*}{Subject6}
			& LR\cite{horikawa2020neural} & 0.489 & 0.437 & 0.305& 0.291&0.375&0.355\\
			& Benchmark & 0.503 & 0.330 & 0.218& 0.196&0.447&0.406\\
			&CA2E\cite{yeh2017learning}&0.562&0.388&0.356&0.344&0.392&0.346\\
			&  ML-GCN\cite{chen2019multi}  &  0.461 & 0.313  & 0.191  &0.174  &0.466 &0.427  \\
			&  SIMM\cite{wu2019multi}  & 0.475 &0.324   & 0.250  & 0.231  &0.454 &0.431  \\
			&  GARDIS\cite{chen2020multi}  & 0.487  &0.341   & 0.133    &0.133   &0.431 &0.401 \\
			\cline{2-8}
			&ML-BVAE  & \textbf{0.427}  & \textbf{0.305}  & \textbf{0.377}  & \textbf{0.367} & \textbf{0.484}& \textbf{0.457} \\
			\hline
			\hline
			\multirow{7}{*}{Subject7}
			& LR\cite{horikawa2020neural} & 0.497 & 0.431 & 0.318& 0.308&0.382&0.364\\
			& Benchmark & 0.463 & 0.316 & 0.252& 0.229&0.465&0.438\\
			&CA2E\cite{yeh2017learning}&0.600&0.399&0.364&0.353&0.382&0.337\\
			&  ML-GCN\cite{chen2019multi} &0.441 & 0.301 & 0.245  &0.230   & 0.482 & 0.460 \\
			&  SIMM\cite{wu2019multi}  &  0.471 & 0.320 & 0.256 &0.239   &0.453 &0.449  \\
			&  GARDIS\cite{chen2020multi}  & 0.491  & 0.335 & 0.154    & 0.154  & 0.444 &0.425 \\
			\cline{2-8}
			&ML-BVAE  & \textbf{0.429}  & \textbf{0.299}  & \textbf{0.379}  & \textbf{0.368} & \textbf{0.489}&\textbf{0.479} \\
			\hline
			\hline
			\multirow{7}{*}{Subject8}
			& LR\cite{horikawa2020neural} & 0.529 & 0.440 & 0.288& 0.273&0.361&0.352\\
			& Benchmark & 0.473 & 0.325 & 0.224& 0.204&0.457&0.414\\
			&CA2E\cite{yeh2017learning}&0.584&0.390&0.335&0.322&0.390&0.342\\
			&  ML-GCN\cite{chen2019multi}  &  0.475 & 0.310  & 0.205  &0.193   &0.466 & 0.435 \\
			&  SIMM\cite{wu2019multi}  &  0.499 & 0.331  & 0.247  & 0.231  &0.448 & 0.435 \\
			&  GARDIS\cite{chen2020multi}  & 0.516  &0.347  & 0.131    &0.132   & 0.426&0.401 \\
			\cline{2-8}
			&ML-BVAE  & \textbf{0.467}  & \textbf{0.303}  & \textbf{0.364}  & \textbf{0.355} & \textbf{0.480}&\textbf{0.454}  \\
			\hline
			\hline
			\multirow{7}{*}{Average}
			& LR\cite{horikawa2020neural} & 0.513 & 0.435 & 0.306& 0.306&0.372&0.360\\
			& Benchmark &0.465  &0.318  &0.229 &0.209 &0.461 &0.427\\
			&CA2E\cite{yeh2017learning}&0.575 &0.392 &0.352&0.339 &0.390 &0.340\\
			&  ML-GCN\cite{chen2019multi}  & 0.455  & 0.306  &  0.223 & 0.208  &0.474 &0.448  \\
			&  SIMM\cite{wu2019multi}  & 0.476  &0.320  &0.258  &0.239   &0.457 & 0.443 \\
			&  GARDIS\cite{chen2020multi}  & 0.502  & 0.339  & 0.147    &  0.147 & 0.439 &0.410 \\
			\cline{2-8}
			&ML-BVAE  & \textbf{0.424}  & \textbf{0.298}  & \textbf{0.378}  & \textbf{0.369}  & \textbf{0.488} & \textbf{0.471} \\
			\hline
			\hline
	\end{tabular}}
	\label{MEMO80_all}
\end{table}
\subsection{Evaluation metrics}
We employ six widely used metrics for measuring multi-label emotion decoding performance, including \textbf{One-Error} (OneE), \textbf{Ranking Loss} (RL), \textbf{Micro F1} (miF1), \textbf{Macro F1} (maF1), \textbf{exampled-based Average Precision} (e-AP) \cite{zhang2013review} and \textbf{mean Average Precision}\footnote{e-AP and mAP are two totally different metrics proposed by different references despite of the similar name. We add an 'e' to the former for disambiguation.} (mAP) \cite{chen2021learning}. For OneE and RL, the smaller the values the better the performance. For the other four metrics, the larger the values the better the performance. In these metrics, OneE, RL and e-AP are example-based metrics which can evaluate the model performance on each test example separately and then return the mean value across the test set, while miF1, maF1 and mAP are label-based metrics which are able to evaluate the model performance on each emotion category separately, and then return the mean value across all emotion categories. Among these, miF1, maF1 and mAP are the primary metrics considering their comprehensiveness.

\subsection{Experimental Results}

\subsubsection{Multi-label emotion decoding performance}
For MEMO27, 10-fold cross-validation is performed for each subject where the mean results of each metric are recorded for all compared approaches. For MEMO80, we follow the setting in \cite{koide2020distinct} leading to 3600 samples in the training dataset and 1800 samples in the test dataset, which can ensure that all the stimuli used in the test dataset have not been seen during training. 

Detailed experimental results are reported in Tables \ref{MEMO27_all} and \ref{MEMO80_all}. Compared with other methods, ML-BVAE shows obvious superiority in both datasets. It outperforms all the compared methods in all six metrics regarding total five subjects in MEMO27 and achieves the best performance in most cases in MEMO80. It is also noticeable that ML-BVAE suppresses all the compared algorithms as far as the primary metrics miF1,  maF1 and mAP among the 39 configurations (13 subjects $\times$ 3 metrics). Specifically, for MEMO27 our method has a relative improvement of 12.3\% and 18.8\% on miF1 and maF1 respectively (five subjects averaged) compared with the second place method and a relative improvement of 19.0\% and 17.1\% can be reached for MEMO80, which is a significant improvement in the neural decoding area.

Among the compared approaches, LR shows the least superiority for multi-label emotion decoding task especially in MEMO27 due to its relatively low representation learning ability and neglect of label correlations. ML-GCN and Benchmark are two relatively strong baselines which can rank 2nd or 3rd for most cases. CA2E shows some superiority with regard to miF1 and maF1. Surprisingly, SIMM and GARDIS, as two multi-view multi-label learning methods, perform worse than other single-view baselines and far worse than ML-BVAE. Between them, GARDIS achieves the worst results perhaps because it is the only method without deep learning. This shows that emotion decoding task needs to learn more expressive neural representations which is consistent with the original intention of our hybrid model design.

\begin{table}[H]
	\centering
	\caption{Friedman statistics $F_F$ in terms of each metric and the critical value at 0.05 significance level. ($\#$ compared algorithms $k=6$, $\#$ subjects $N=13$ for MEMO27 and MEMO80.)}
	\resizebox{\columnwidth}{!}{
		\begin{tabular}{cccccc|c} 
			\hline
			\hline
			OneE &RL& miF1 & maF1&e-AP&mAP& critical value\\
			\hline	
			73.306 &243.328& 40.616& 60.651&106.300&118.719& 2.368\\
			\hline
			\hline
	\end{tabular}}
	\label{friedman}
\end{table}
\begin{figure*}
	\centering
	\subfigure[ OneE] {\includegraphics[height=0.6in,width=2.3in]{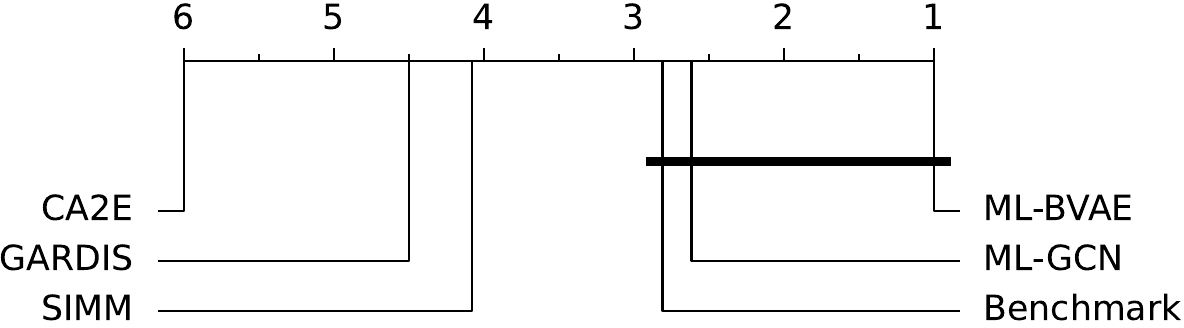}}
	\subfigure[ RL] {\includegraphics[height=0.6in,width=2.3in]{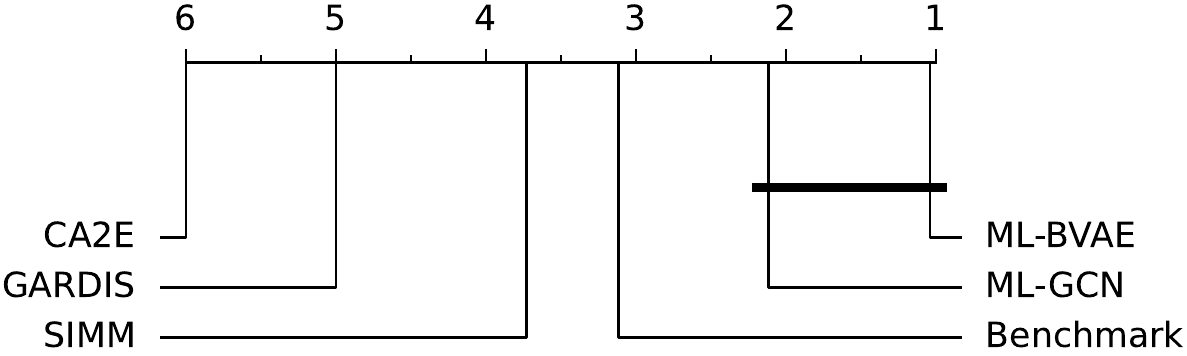}}
	\subfigure[ miF1] {\includegraphics[height=0.6in,width=2.3in]{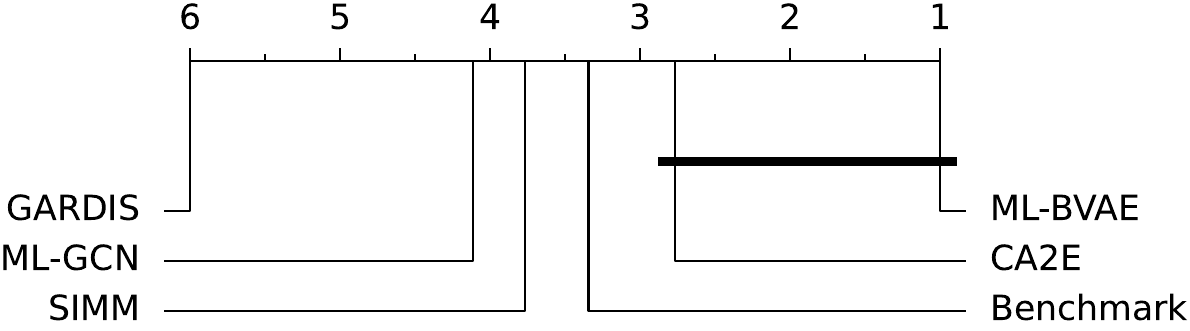}}
	\subfigure[ maF1] {\includegraphics[height=0.6in,width=2.3in]{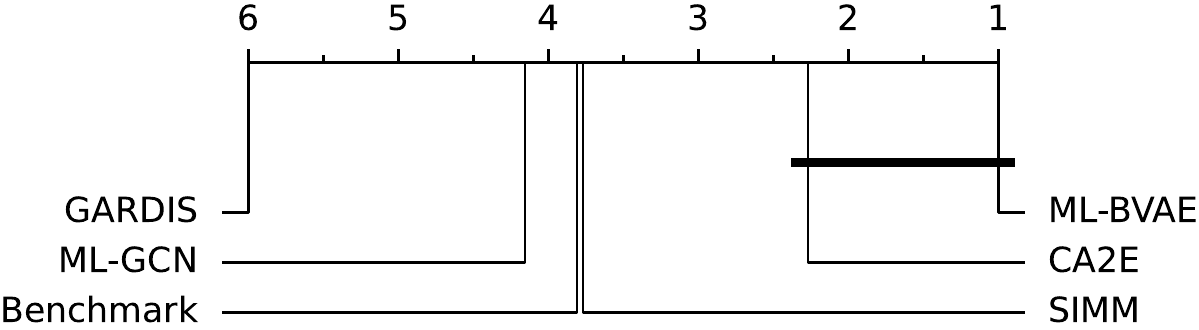}}
	\subfigure[ e-AP] {\includegraphics[height=0.6in,width=2.3in]{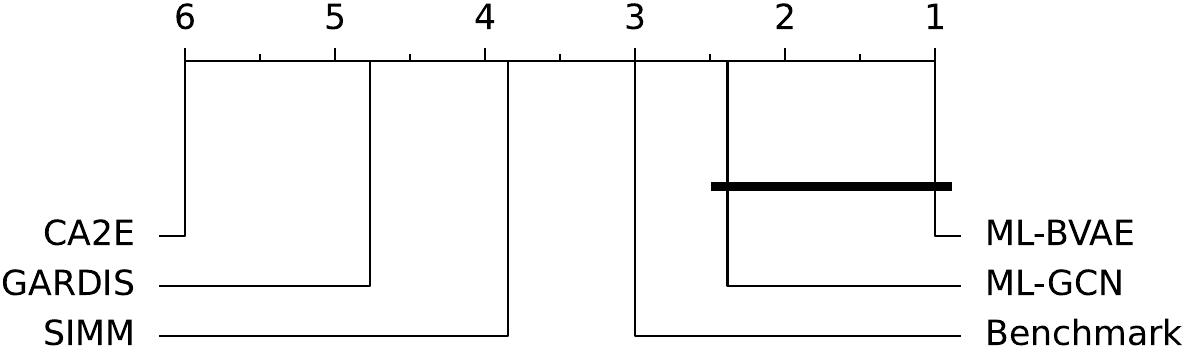}}
	\subfigure[ mAP] {\includegraphics[height=0.6in,width=2.3in]{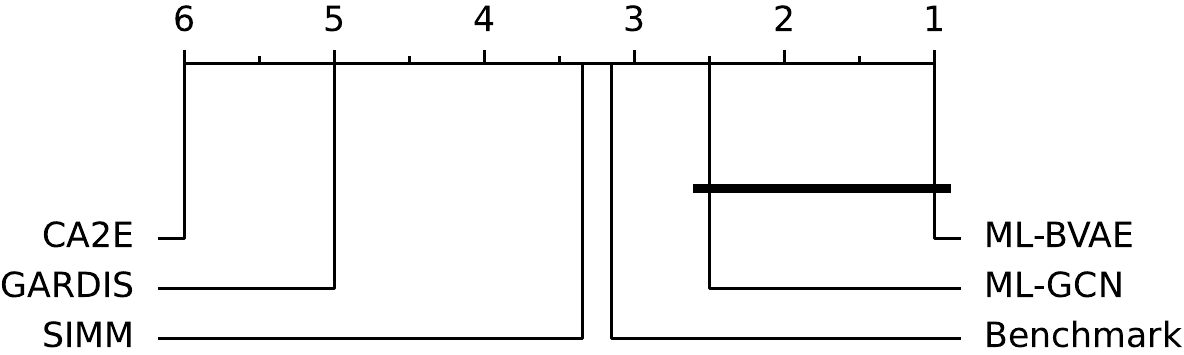}}
	\caption{Comparisons of ML-BVAE (control algorithm) against other comparing algorithms with the Bonferroni-Dunn test in 13 subjects of MEMO27 and MEMO80. Algorithms not connected with ML-BVAE in the CD diagram are considered to have significantly different performance from the control algorithm (CD=1.890 at 0.05 significance level).}
	\label{bd}
	\vspace{-0.5cm}
\end{figure*}

Furthermore, we adopt the Friedman test \cite{demvsar2006statistical} for statistical test in order to discuss the relative performance among the compared methods. If there are $k$ algorithms and $N$ datasets (the number of subjects in our experiment), we take use of the average ranks of algorithms $R_j = \frac{1}{N}\sum_{i}r_i^j$ for Friedman test in which $r_i^j$ is the ranks of the $j$-th algorithm on the $i$-th dataset. If the null-hypothesis is that all the algorithms have the equivalent performance, the Friedman statistic $F_F$ which will satisfy the F-distribution with $k-1$ and $(k-1)(N-1)$ degrees of freedom can be written as:
\begin{equation}
	F_F = \frac{(N-1)\chi_F^2}{N(k-1)-\chi_F^2},
\end{equation}
in which
\begin{equation}
	\chi_F^2 = \frac{12N}{k(k+1)}\Bigg[\sum_{j=1}^{k}R_j^2-\frac{k(k+1)^2}{4}\Bigg].
\end{equation}
Table \ref{friedman} shows the Fridman statistics $F_F$ and the corresponding critical value in regard to each metric ($\#$ comparing algorithms $k=6$ (except \textbf{LR})\footnote{LR is not considered in Friedman test and post-hoc Bonferroni-Dunn test since it is not a single model but one model per emotion category. We only explore multi-label learning methods in this section.}, $\#$ subjects $N=13$). With respect to each metric, the null hypothesis of equivalent performance among the compared methods can be rejected at the 0.05 significance level.

Then, we perform the strict post-hoc Bonferroni-Dunn test \cite{demvsar2006statistical} which is used to account for the relative performance between ML-BVAE (control algorithm) and other compared approaches. The critical difference (CD) value of the rank difference between two algorithms is: 
\begin{equation}
	CD=q_{\alpha}\sqrt{\frac{k(k+1)}{6N}},
\end{equation}
in which $q_{\alpha}=2.576$ at 0.05 significance level. Therefore, ML-BVAE can be considered as having significantly different performance than a compared algorithm if their average ranks difference is larger than CD (CD=1.890 in our experimental setting). Fig. \ref{bd} reports the CD diagrams on each metric, where the average rank of each compared method is marked along the axis (the smaller the better). Algorithms not connected with ML-BVAE in the CD diagram are considered to have significantly different performance from the control algorithm. We can observe that: (1) ML-BVAE achieves the best average rank with respect to all metrics. (2) As far as the primary metrics, ML-BVAE is significantly better than the compared methods other than CA2E in terms of miF1 and maF1 and achieves the significant best results in mAP compared with other methods rather than ML-GCN. (3) Although some strong baselines such as ML-GCN are not significantly different from ML-BVAE in terms of a few metrics which are of little importance, our method still has obvious superiority when all metrics are taken into consideration and achieves a consistently better average ranks. These experimental results convincingly illustrate the significance of the superiority of our ML-BVAE approach.
\begin{table*}
	\centering
	\caption{Comparisons of AP and mAP in $\%$ of ML-BVAE and compared methods in MEMO27 (five subjects averaged). A: Admiration; B: Adoration; C: Aesthetic Appreciation; D: Amusement; E: Anger; F: Anxiety; G: Awe; H: Awkwardness; I: Boredom; J: Calmness; K: Confusion; L: Craving; M: Disgust; N: Empathic Pain; O: Entrancement; P: Excitement; Q: Fear; R: Horror; S: Interest; T: Joy; U: Nostalgia; V: Relief; W: Romance; X: Sadness; Y: Satisfaction; Z: Sexual Desire; $\Omega$: Surprise.}
	\vspace{-0.1cm}
	\tiny
	\resizebox{\textwidth}{!}{
		\begin{tabular}{@{\,}c@{\,} *{27}{@{\,}c@{\,}}@{\,}c@{\,}}
			\hline
			\hline
			Method&A &B &C &D &E &F &G &H &I &J &K &L &M &N &O &P &Q &R &S &T &U&V&W&X&Y&Z & $\Omega$ &mAP\\
			\hline
			\rowcolor{mygray}[1.5pt][1.5pt]
			LR\cite{horikawa2020neural}&  16.7  & 18.8  & 45.5  & 58.7  & 3.1   & 33.3  & 35.7  & 6.2   & 17.9  & 8.9   & 15.8  & 8.6   & 42.4  & 12.2  & 23.7  & 20.6  & 38.3  & 32.4  & 36.8  & 26.5  & 17.5  & 7.8   & 18.7  & 14.2  & 17.7  & 50.7  & 35.8 &24.6  \\
			Benchmark& \textbf{26.1}  & 51.6  & 66.7  & 76.1  & \textbf{24.1}  & 51.8  & 59.4  & \textbf{23.3}  & 26.3  & 32.1  & 33.1  & 27.9  & 57.4  & 32.0  & 51.0  & 31.8  & 57.8  & 51.5  & 59.2  & 53.5  & 28.7  & 22.3  & 41.5  & 34.2  & 26.3  & 70.0  & 56.4  & 43.4  \\
			\rowcolor{mygray}[1.5pt][1.5pt]
			CA2E\cite{yeh2017learning}&20.3  & 37.4  & 50.7  & 66.8  & 16.3  & 40.4  & 49.8  & 17.9  & 21.8  & 23.4  & 26.8  & 21.2  & 42.7  & 23.6  & 37.4  & 24.0  & 44.1  & 39.7  & 51.3  & 41.4  & 22.0  & 20.1  & 31.2  & 26.0  & 21.3  & 57.5  & 45.8  & 34.1  \\
			ML-GCN\cite{chen2019multi} & 25.7  & 51.2  & 66.4  & 76.5  & 19.9  & 51.9  & 58.7  & 21.4  & 25.5  & 29.8  & 34.2  & 25.0  & 56.6  & 29.2  & 50.6  & 30.1  & 57.2  & 50.4  & 61.1  & 52.7  & 27.7  & 19.5  & 36.9  & 32.4  & 25.3  & 65.6  & 56.8  & 42.2  \\
			\rowcolor{mygray}[1.5pt][1.5pt]
			SIMM\cite{wu2019multi}  & 25.0  & 42.1  & 61.2  & 73.8  & 19.1  & 50.3  & 54.4  & 21.1  & 24.2  & 30.2  & 32.0  & 15.2  & 54.1  & 30.8  & 47.7  & 28.2  & 53.6  & 46.9  & 59.3  & 47.3  & 27.6  & 22.3  & 33.1  & 30.4  & 26.1  & 65.3  & 54.2  & 39.8  \\
			GARDIS\cite{chen2020multi}  & 19.3  & 43.4  & 62.9  & 73.8  & 11.7  & 47.4  & 54.9  & 15.8  & 20.4  & 28.1  & 25.4  & 24.2  & 53.5  & 25.3  & 48.0  & 22.8  & 52.5  & 44.7  & 61.0  & 48.3  & 26.1  & 14.4  & 31.8  & 26.1  & 19.0  & 67.2  & 55.1  & 37.9  \\
			\hline
			\rowcolor{mygray}[1.5pt][1.5pt]
			ML-BVAE & 25.9 & \textbf{53.3} & \textbf{68.4} & \textbf{76.9} & 23.5 & \textbf{55.1} & \textbf{60.2} & 22.9 & \textbf{26.7} & \textbf{33.7} & \textbf{34.7} & \textbf{29.8} & \textbf{58.4} & \textbf{33.5} & \textbf{52.6} & \textbf{32.1} & \textbf{59.9} & \textbf{53.4} & \textbf{61.4} & \textbf{54.7} & \textbf{30.7} & \textbf{22.9} & \textbf{45.4} & \textbf{36.0} & \textbf{26.4} & \textbf{73.9} & \textbf{57.8 } & \textbf{44.8} \\
			\hline
			\hline		
	\end{tabular}}
	\vspace{-0.2cm}
	\label{map1}
\end{table*}
\begin{figure}
	\centering
	\includegraphics[height=1.7in,width=3.5in]{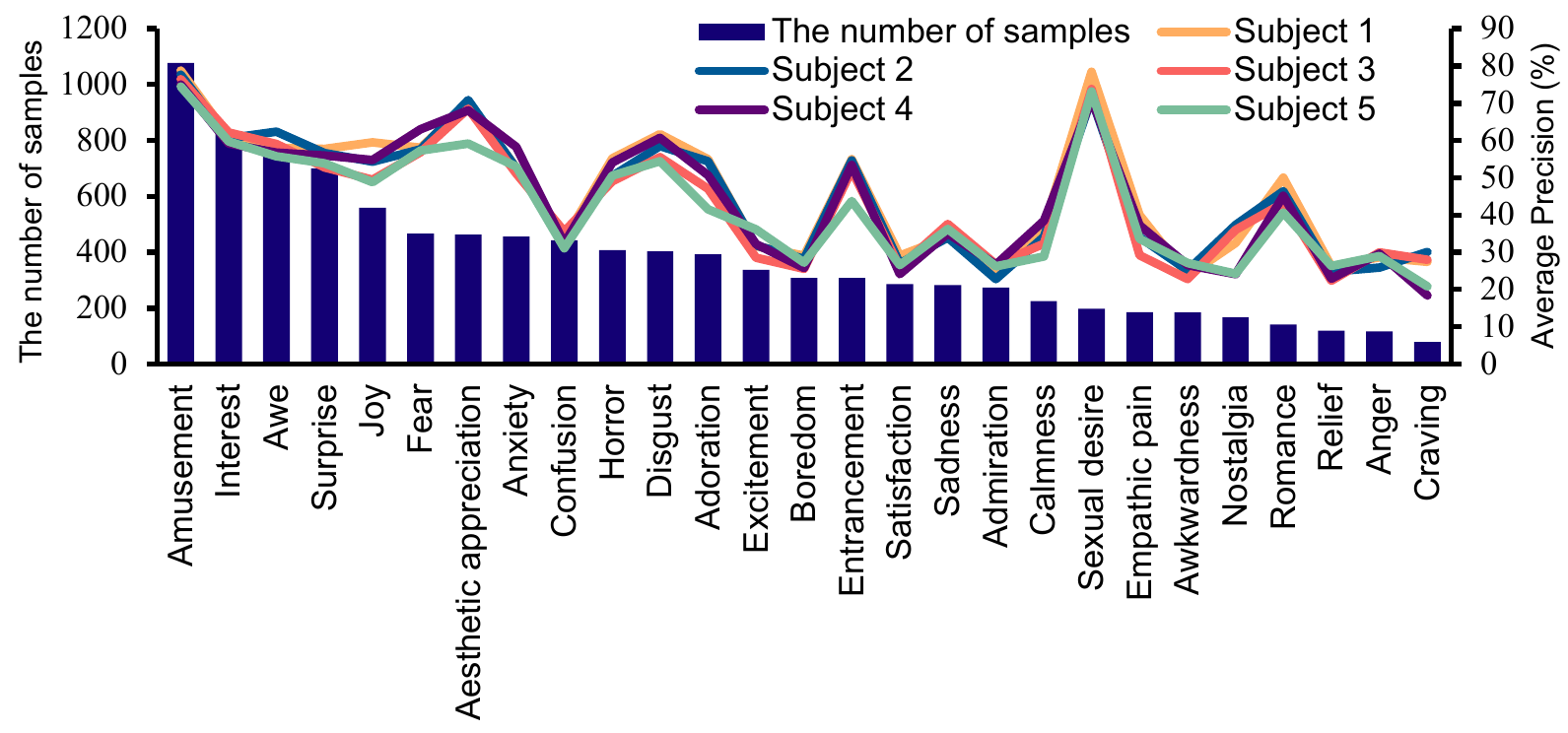}
	\caption{A bar graph of the number of samples corresponding to each emotion category and line graphs of each subject's emotional decoding AP in MEMO27.}
	\vspace{-0.2cm}
	\label{ld_27}
	\vspace{-0.2cm}
\end{figure}
\begin{table*}[t]
	\centering
	\caption{Comparisons of AP and mAP in $\%$ of ML-BVAE and compared methods in MEMO80 (eight subjects averaged). 1.Love; 2.Amusement; 3.Craving; 4.Joy; 5.Nostalgia;	6.Boredom; 7.Calmness; 8.Relief; 9.Romance;	10.Sadness;	11.Admiration; 12.Aesthetic Appreciation; 13.Awe;	14.Confusion; 15.Entrancement; 16.Interest; 17.Satisfaction; 18.Excitement; 19.Sexual Desire; 20.Surprise; 21.Nervousness; 22.Tension; 23.Anger; 24.Anxiety; 25.Awkwardness; 26.Disgust; 27.Empathic Pain; 28.Fear; 29.Horror; 30.Laughing;	31.Happiness; 32.Friendliness; 33.Ridiculousness; 34.Affection; 35.Liking; 36.Shedding Tears; 37.Emotional Hurt; 38.Sympathy; 39.Lethargy; 40.Empathy; 41.Compassion; 42.Curiousness; 43.Unrest; 44.Exuberance; 45.Appreciation of Beauty;	46.Fever; 47.Scare; 48.Daze; 49.Positive-Expectation; 50.Throb; 51.Sexiness; 52.Indecency; 53.Embarrassment; 54.Oddness; 55.Contempt; 56.Alertness; 57.Eeriness; 58.Positive-Emotion; 59.Vigor; 60.Longing; 61.Tenderness; 62.Pensiveness; 63.Melancholy; 64.Relaxedness; 65.Acceptance;	66.Unease; 67.Negative-Emotion; 68.Hostility; 69.Levity; 70.Protectiveness; 71.Elation; 72.Coolness; 73.Cuteness; 74.Attachment; 75.Encouragement; 76.Annoyance; 77.Positive-Fear; 78.Aggressiveness; 79.Distress; 80.Stress.}
	\vspace{-0.1cm}
	\tiny
	\resizebox{\textwidth}{!}{
		\begin{tabular}{@{\,}c@{\,} *{27}{@{\,}c@{\,}}}
			\hline
			\hline
			Method&1 &2 &3 &4 &5 &6 &7 &8 &9 &10 &11 &12 &13 &14 &15 &16 &17 &18 &19 &20 &21&22&23&24&25&26 & 27 \\
			\hline
			\rowcolor{mygray}[1.5pt][1.5pt]
			LR\cite{horikawa2020neural}& 36.7  & 42.3  & 25.8  & 39.6  & 46.3  & 44.9  & 26.3  & 35.2  & 26.3  & 35.1  & 15.4  & 34.0  & 43.9  & 42.7  & 47.1  & 47.1  & 37.9  & 44.2  & 39.9  & 66.6  & 42.7  & 38.7  & 16.8  & 50.9  & 37.9  & 38.5  & 34.4  \\

			Benchmark& 43.7  & 48.7  & 27.8  & 50.1  & 59.5  & 45.9  & \textbf{35.9}  & 42.2  & 32.0  & 33.0  & \textbf{17.8}  & 42.6  & 54.1  & 54.5  & 41.3  & 52.6  & 39.9  & 50.7  & 29.4  & 69.6  & 52.5  & 46.0  & 26.2  & 59.5  & 26.4  & 54.7  & 41.7  \\
			\rowcolor{mygray}[1.5pt][1.5pt]
			CA2E\cite{yeh2017learning}&36.0  & 43.1  & 28.6  & 47.2  & 44.0  & 34.5  & 30.0  & 36.2  & 28.4  & 33.1  & 14.9  & 40.1  & 40.3  & 38.0  & 43.0  & 45.9  & 36.6  & 37.7  & 28.4  & 58.8  & 34.6  & 30.0  & 21.3  & 40.7  & 30.9  & 34.5  & 30.3  \\
			
			ML-GCN\cite{chen2019multi} & 43.5  & 50.5  & 30.1  & 49.3  & 61.9  & 47.7  & 35.7  & 42.0  & 32.7  & 36.8  & 16.9  & 42.6  & 57.8  & 55.3  & 48.2  & 53.8  & 39.7  & 55.6  & 50.5  & \textbf{70.3}  & 55.0  & 47.3  & 26.5  & 60.9  & 41.3  & 55.2  & 43.4  \\
			\rowcolor{mygray}[1.5pt][1.5pt]
			SIMM\cite{wu2019multi}  & 39.3  & 49.7  & 31.8  & 48.8  & 61.9  & 49.8  & 33.3  & 41.4  & 30.9  & 37.3  & 16.3  & 41.7  & 56.3  & 54.0  & 48.9  & 52.2  & 37.7  & 57.1  & 56.1  & 69.8  & 54.7  & 45.7  & 25.8  & 58.0  & 44.8  & 53.5  & 44.4  \\
			GARDIS\cite{chen2020multi} &35.9  & 47.4  & 30.6  & 44.8  & 54.7  & 45.5  & 30.7  & 40.3  & 29.4  & 33.4  & 14.9  & 40.8  & 52.1  & 48.8  & 46.6  & 51.2  & 36.5  & 52.2  & 55.2  & 67.3  & 49.2  & 40.2  & 21.4  & 53.0  & 38.5  & 48.5  & 43.7  \\
			\hline
			\rowcolor{mygray}[1.5pt][1.5pt]
			ML-BVAE & \textbf{44.7}  & \textbf{51.9}  & \textbf{36.6}  & \textbf{50.5}  & \textbf{62.6}  & \textbf{53.0}  & 34.8  & \textbf{43.7}  & \textbf{34.1}  & \textbf{39.9}  & 17.7  & \textbf{43.6}  & \textbf{59.8}  & \textbf{56.3}  & \textbf{50.9}  & \textbf{54.1}  & \textbf{40.5}  & \textbf{60.2}  & \textbf{61.1}  & \textbf{70.3}  & \textbf{56.3}  & \textbf{48.0}  & \textbf{28.8}  & \textbf{61.2}  & \textbf{45.2}  & \textbf{57.8}  & \textbf{46.3}  \\
			\hline
			\hline		
			Method&28 &29  &30    &31    &32    &33    &34    &35    &36    &37    &38    &39    &40    &41    &42    &43    &44    &45    &46    &47    &48    &49    &50    &51    &52    &53    &54  \\
			\hline
			\rowcolor{mygray}[1.5pt][1.5pt]
			LR\cite{horikawa2020neural}& 41.5  & 42.3  & 23.9  & 30.6  & 31.8  & 36.0  & 23.3  & 26.5  & 24.3  & 38.6  & 43.7  & 50.5  & 29.2  & 32.0  & 31.1  & 38.9  & 42.1  & 24.8  & 18.9  & 40.4  & 46.1  & 47.7  & 40.6  & 40.8  & 52.7  & 31.7  & 26.7  \\
			Benchmark& 56.9  & 56.9  & 19.2  & 37.6  & \textbf{41.1}  & 38.8  & 34.0  & 35.7  & 28.0  & 35.4  & 49.5  & 55.8  & 37.7  & 46.0  & 26.2  & 55.3  & 48.2  & 34.5  & 30.2  & 51.8  & 55.3  & 50.5  & 48.3  & 36.4  & 53.0  & 31.7  & 32.3  \\
			\rowcolor{mygray}[1.5pt][1.5pt]
			CA2E\cite{yeh2017learning}&33.6  & 38.9  & 20.8  & 30.6  & 31.2  & 30.8  & 28.6  & 29.4  & 25.4  & 37.5  & 40.0  & 56.8  & 30.1  & 31.1  & 28.7  & 38.0  & 39.4  & 29.7  & 27.3  & 34.9  & 41.2  & 46.7  & 34.2  & 35.9  & 43.1  & 26.2  & 25.9  \\
			ML-GCN\cite{chen2019multi} & 58.3  & 58.5  & 26.8  & 37.9  & 40.5  & 41.6  & 32.9  & 33.7  & 31.3  & 39.7  & 50.3  & 59.1  & 37.9  & 45.7  & 34.7  & 56.7  & \textbf{50.5}  & 33.5  & \textbf{33.3}  & 52.6  & 56.1  & 52.9  & 50.1  & 55.1  & 57.0  & 40.6  & 32.7\\	
			\rowcolor{mygray}[1.5pt][1.5pt]	
			SIMM \cite{wu2019multi} & 57.5  & 57.9  & 27.3  & 37.3  & 38.0  & 41.6  & 30.1  & 33.0  & 29.9  & 41.1  & 49.5  & \textbf{59.2}  & 35.7  & 46.0  & 37.6  & 55.7  & 48.2  & 31.6  & 31.8  & 51.6  & 55.0  & 52.9  & 48.2  & 59.2  & 59.2  & 43.2  & 32.5  \\
			GARDIS\cite{chen2020multi} & 50.2  & 53.2  & 24.0  & 34.9  & 34.4  & 40.5  & 28.3  & 31.9  & 30.1  & 39.3  & 46.7  & 59.1  & 33.6  & 35.9  & 34.8  & 51.1  & 44.5  & 29.6  & 29.2  & 47.9  & 52.0  & 50.3  & 44.5  & 57.7  & 56.3  & 44.0  & 29.4  \\
			\hline
			\rowcolor{mygray}[1.5pt][1.5pt]
			ML-BVAE & \textbf{60.7}  & \textbf{60.7}  & \textbf{31.7}  & \textbf{38.8}  & 40.7  & \textbf{49.4}  & \textbf{34.4}  & \textbf{36.5}  & \textbf{38.1}  & \textbf{47.5}  & \textbf{52.7}  & 57.8  & \textbf{39.4}  & \textbf{46.4}  & \textbf{37.8}  & \textbf{58.1}  & 49.3  & \textbf{35.0}  & 32.8  & \textbf{53.4}  & \textbf{57.2}  & \textbf{53.6}  & \textbf{51.9}  & \textbf{63.8}  & \textbf{61.9}  & \textbf{51.7}  & \textbf{35.3}  \\
			\hline
			\hline		
			Method&55    & 56    & 57    & 58    & 59    & 60    & 61    & 62    & 63    & 64    & 65    & 66    & 67    & 68    & 69    & 70    & 71    & 72    & 73    & 74    & 75    & 76    & 77    & 78    & 79    & 80    &mAP\\
			\hline
			\rowcolor{mygray}[1.5pt][1.5pt]
			LR\cite{horikawa2020neural}& 29.3  & 24.6  & 41.7  & 36.0  & 40.2  & 32.9  & 41.6  & 28.7  & 32.3  & 36.6  & 26.6  & 32.9  & 41.8  & 36.3  & 38.4  & 27.7  & 19.1  & 41.1  & 22.1  & 20.9  & 28.4  & 40.1  & 44.3  & 38.8  & 36.4  & 45.5&36.0  \\
			
			Benchmark& 38.3  & 37.0  & 57.0  & 40.4  & 51.5  & 35.6  & 50.2  & 35.1  & 38.5  & 47.0  & \textbf{37.0}  & 35.5  & 59.8  & 38.9  & 44.7  & 22.8  & 27.1  & 40.4  & \textbf{28.9}  & 27.2  & 38.8  & \textbf{57.5}  & 59.9  & \textbf{53.0}  & 47.2  & 60.6  & 42.7  \\
			\rowcolor{mygray}[1.5pt][1.5pt]
			CA2E\cite{yeh2017learning}&33.3  & 18.5  & 31.8  & 36.5  & 45.4  & 31.1  & 42.1  & 31.2  & 27.9  & 41.6  & 28.7  & 27.1  & 39.6  & 31.3  & 33.9  & 25.2  & 22.4  & 36.3  & 23.9  & 21.2  & 33.4  & 32.7  & 36.1  & 37.4  & 32.7  & 39.0  & 34.0  \\	
			ML-GCN\cite{chen2019multi} & 39.8  & 38.6  & 58.4  & 42.3  & 51.5  & 34.1  & 50.0  & 36.7  & 39.4  & 47.0  & 36.2  & \textbf{37.4}  & 61.5  & 43.1  & 47.7  & 30.0  & 27.8  & 43.5  & 28.3  & \textbf{27.7}  & 40.0  & 56.5  & \textbf{61.7}  & 51.9  & 47.1  & \textbf{61.8}  & 44.9  \\
			\rowcolor{mygray}[1.5pt][1.5pt]
			SIMM \cite{wu2019multi} & 40.8  & 38.0  & 55.9  & 42.8  & 49.6  & 33.2  & 46.2  & 36.0  & 38.9  & 46.2  & 33.1  & 36.9  & 60.4  & 45.6  & 49.7  & 31.9  & 25.7  & 43.4  & 26.2  & 25.0  & 39.2  & 53.1  & 59.7  & 49.8  & 46.0  & 59.5  & 44.3  \\
			GARDIS\cite{chen2020multi} &  37.0  & 35.1  & 48.9  & 40.7  & 46.9  & 33.0  & 43.5  & 34.6  & 33.9  & 43.3  & 30.6  & 32.9  & 55.0  & 37.8  & 46.8  & 28.1  & 23.9  & 40.1  & 23.8  & 23.4  & 36.9  & 46.9  & 53.5  & 44.9  & 41.8  & 54.5  & 41.0  \\
			\hline
			\rowcolor{mygray}[1.5pt][1.5pt]
			ML-BVAE & \textbf{41.8}  & \textbf{41.2}  & \textbf{58.5}  & \textbf{45.1}  & \textbf{52.0}  & \textbf{35.8}  & \textbf{52.4}  & \textbf{40.0}  & \textbf{39.8}  & \textbf{48.4}  & 36.7  & \textbf{37.4}  & \textbf{62.6}  & \textbf{48.0}  & \textbf{51.7}  & \textbf{35.1}  & \textbf{30.1}  & \textbf{46.6}  & 28.1  & 27.2  & \textbf{43.1}  & 57.0  & 61.2  & 52.1  & \textbf{49.0}  & 61.5  & \textbf{47.1}  \\
			\hline
			\hline		
	\end{tabular}}
	\label{map2}
\end{table*}

\begin{figure*}[h]
	\centering
	\includegraphics[height=1.5in,width=7.2in]{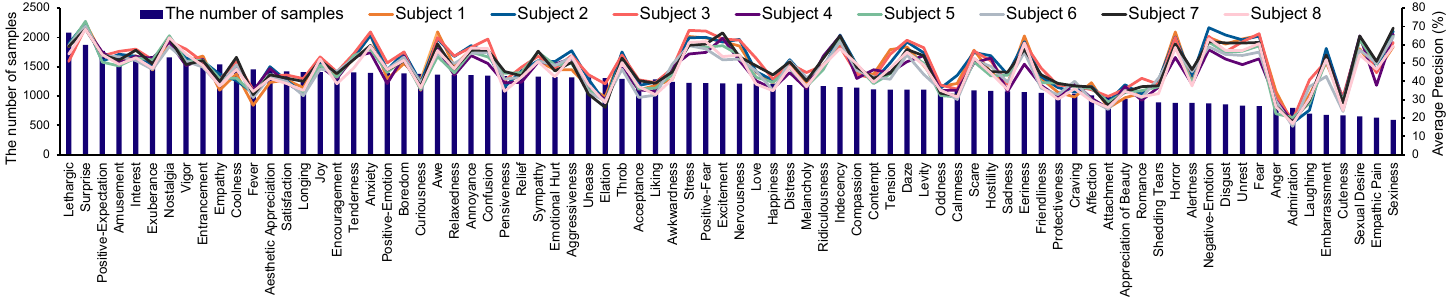}
	\caption{A bar graph of the number of samples corresponding to each emotion category and line graphs of each subject's emotional decoding AP in MEMO80.}
	\label{ld_80}
	\vspace{-0.6cm}
\end{figure*}
\subsubsection{Average precision across all categories}
We also report the average precision\footnote{This average precision used in this section refers to the intermediate result when calculate the label-based metric mAP.} across all emotion categories with respect to MEMO27 and MEMO80 in Tables \ref{map1} and \ref{map2}. The results for each specific emotion label can be found in the bar graphs of Appendix B. We can observe that ML-BVAE outperforms other compared methods in most emotion categories in both datasets. To be specific, ML-BVAE ranks 1st and 2nd in 92.6\% (25/27) and 7.4\% (2/27) respectively in MEMO27, and 83.8\% (67/80) and 13.8\% (11/80) respectively in MEMO80. To further analysis the decoding performance of ML-BVAE in each emotion category, we show the average precision corresponding to each emotional state in order of sample size for each category in Figs. \ref{ld_27} and \ref{ld_80}. Based on these results, we can make the following observations: (1) The experimental results show a strong consistency on multiple subjects in each dataset which indicates that the decoding performance of ML-BVAE is stable to some extent in the face of individual differences in fMRI data. (2) The number of samples is a very important factor that affects the decoding performance. MEMO27 suffers from the long tail distribution of labels leading to lower decoding performance in the categories with a smaller number of samples in general. (3) The intensity of emotional stimuli is another factor which has an effect on emotion decoding. For certain strong emotions such as \emph{Sexual desire} and \emph{Romance} in MEMO27 and \emph{Sexiness} in MEMO80, they still have high decoding accuracy despite of small samples. (4) The label quality also exerts influence on the decoding accuracy. For example, in MEMO80 \emph{Lethargic} and \emph{Fever} have the lowest consistency in emotion ratings among annotators \cite{koide2020distinct}, in which \emph{Fever} has very low decoding accuracy while \emph{Lethargic} has not significant bad performance due to large sample size. We can also find in Fig. \ref{attention} that \emph{Confusion} in MEMO27 and \emph{Admiration} in MEMO80 are both noisy labels because they have no obvious co-occurrence with other labels, which leads to low accuracy. (5) Although the number of decoding categories of MEMO80 is much greater than that of MEMO27 leading to a harder decoding problem, it still achieves considerable decoding accuracy in terms of mAP. This may be due to the higher quality of fMRI data thanks to extra auditory stimulation in the movies of MEMO80.
\begin{figure}
	\centering
	\subfigure[ Masked self-attention martix in MEMO27.] {\includegraphics[height=1.4in,width=1.7in]{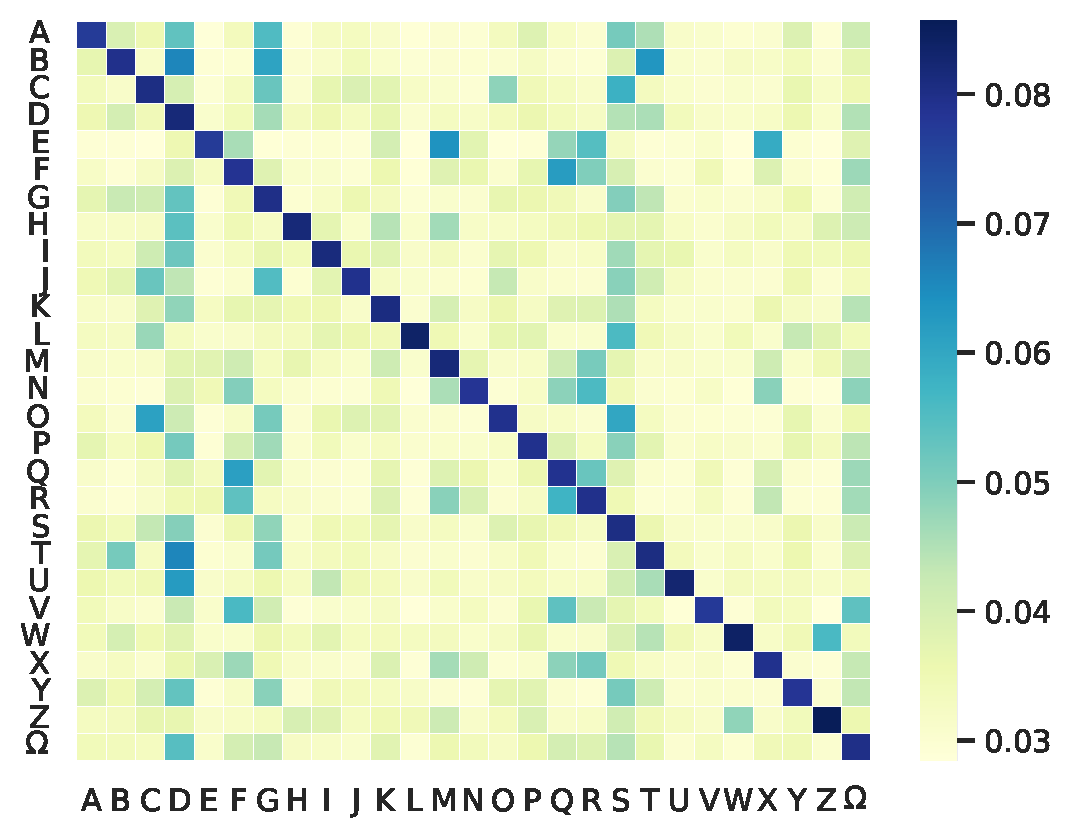}}
	\subfigure[ Masked self-attention martix in MEMO80.] {\includegraphics[height=1.4in,width=1.7in]{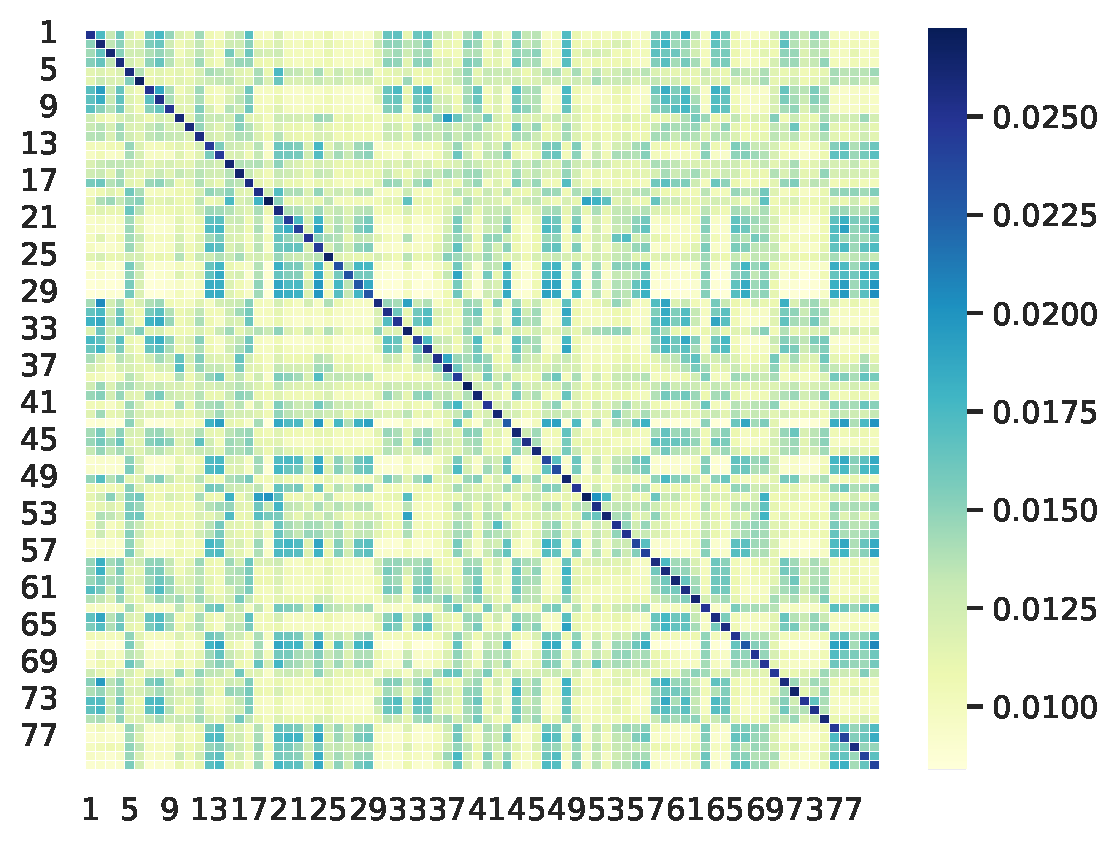}}
	\caption{Visualization of the masked self-attention matrix in MEMO27 and MEMO80. The indexes of emotion labels in (a) and (b) are the same as Figs. \ref{map1} and \ref{map2}.}
	\label{attention}
\end{figure}
\begin{figure}
	\centering
	\subfigure[ \emph{Adoration} in MEMO27.]{\includegraphics[height=0.35in,width=3.5in]{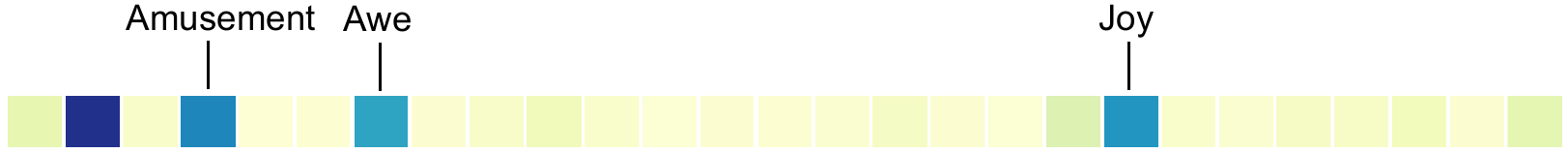}}
	\subfigure[ \emph{Anger} in MEMO27.]{\includegraphics[height=0.55in,width=3.5in]{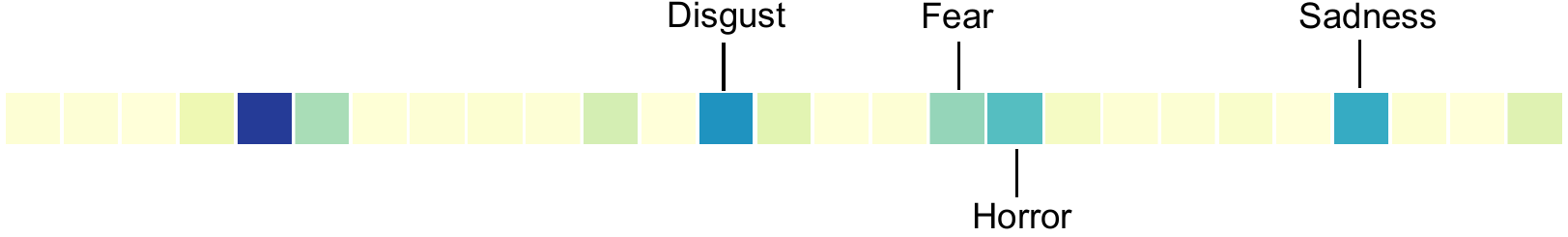}}
	\subfigure[ \emph{Love} in MEMO80.]{\includegraphics[height=0.51in,width=3.5in]{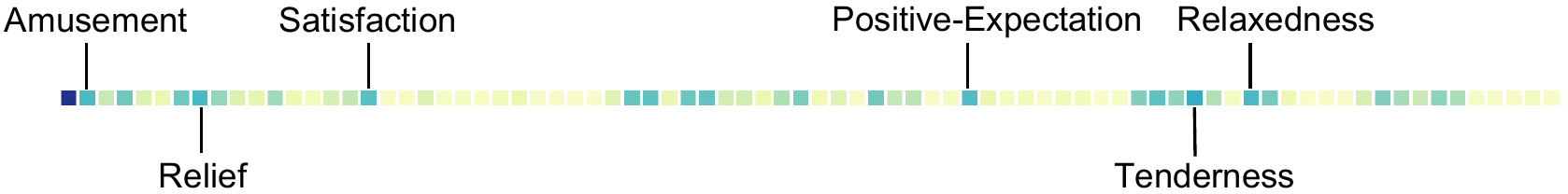}}
	\subfigure[ \emph{Nervous} in MEMO80.]{\includegraphics[height=0.55in,width=3.5in]{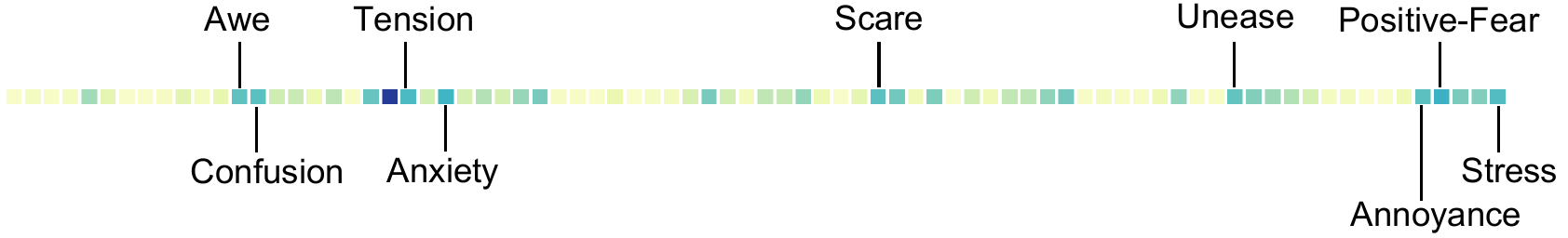}}
	\caption{Visualization of a few rows of the masked self-attention matrix. Each element can be regarded as the probability of co-occurrence with the emotion corresponding to the row.}
	\label{attention_bar}
	\vspace{-0.4cm}
\end{figure}
\begin{figure}
	\centering
	\subfigure[ MEMO27.] {\includegraphics[height=1.5in,width=3.5in]{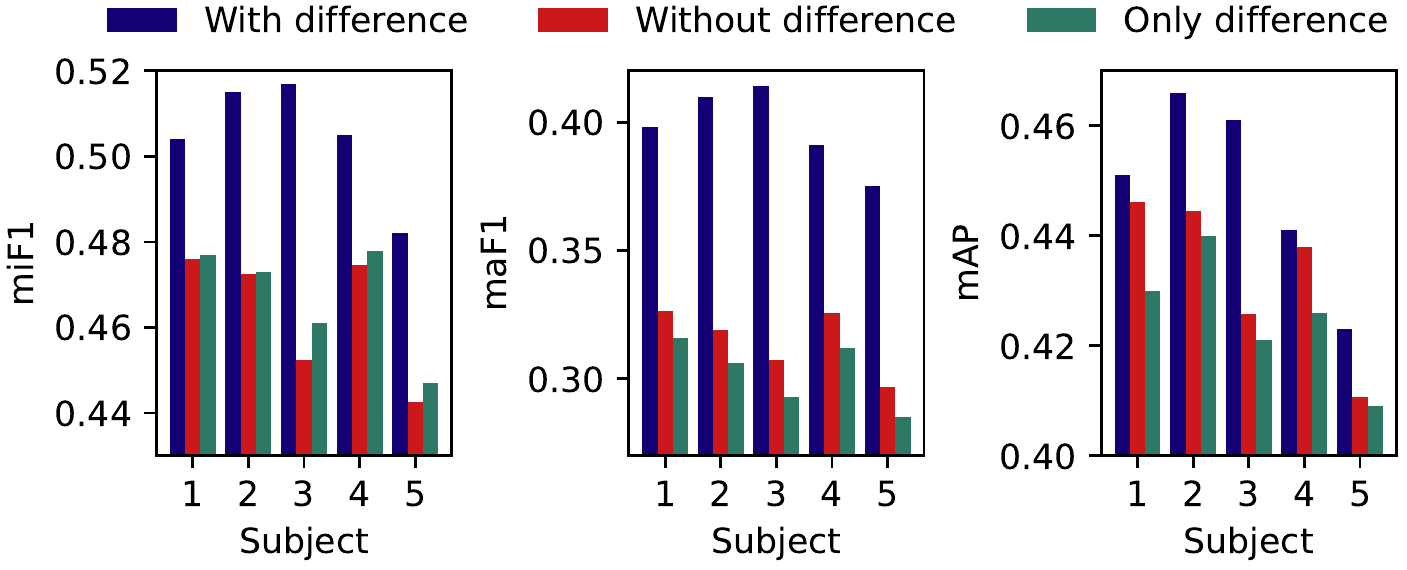}}
	\subfigure[ MEMO80.] {\includegraphics[height=1.3in,width=3.5in]{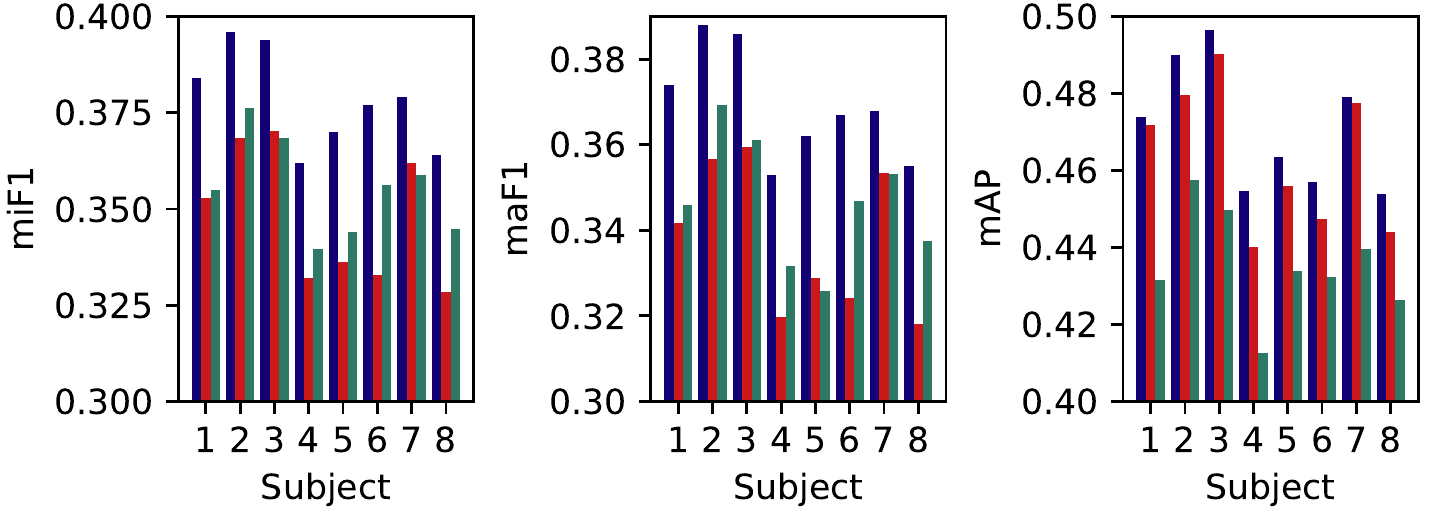}}
	\caption{Multi-label emotion decoding performance in terms of the primary metrics Micro F1, Macro F1 and mean Average Precision in MEMO27 and MEMO80 when the difference view is included (blue bar), excluded (red bar) and only input (green bar).}
	\label{2v3vbar}
	\vspace{-0.5cm}
\end{figure}
\subsubsection{Masked self-attention visualization}
We present the masked self-attention matrix $\mathbf{A}$ in Eq. (\ref{attention_mat}) which is adopted in the multi-label classification network for modeling the emotion label dependency in Fig. \ref{attention}. The attention matrix can be seen as a label-wise correlation map, whose primary component is the fixed label co-occurrence mask. A learnable component is also developed which can adjust the label correlation map slightly. The darker the color in the matrix, the greater the value of the attention, which is also deemed to be the greater the probability that the two emotions can be stimulated simultaneously. To be specific, we show the corresponding rows of several common emotions in the attention matrix in Fig. \ref{attention_bar}. We select 0.5 and 0.6 (before softmax) as the thresholds for MEMO27 and MEMO80 to mark the elements, respectively. For example, in MEMO27, if we express \emph{Adoration} we will be more likely to express \emph{Amusement}, \emph{Awe} and \emph{Joy}; while if we feel \emph{Angry} we may also express \emph{Disgust}, \emph{Fear}, \emph{Horror} and \emph{Sadness}. Similar results can be observed in Fig. \ref{attention_bar} (c) and (d) regarding MEMO80. 
\subsubsection{The effectiveness of bi-hemisphere discrepancy}
To further prove that the bi-hemisphere discrepancy can provide additional information for emotion decoding which leads to better decoding performance, we compare the two views in which the difference view is not included and all three views as the input on all subjects in MEMO27 and MEMO80. Fig. \ref{2v3vbar} reports the corresponding results on the primary metrics miF1, maF1 and mAP. We can observe that, compared with the experiments without the difference view, the decoding performance on all subjects in both datasets has been improved. It's worth noting that the improvement of miF1 and maF1 in MEMO27 is particularly significant when bi-hemisphere discrepancy is considered.

Apart from this, Fig. \ref{2v3vbar} also includes the performance when our model trained on only the difference view. For fair comparison, we realize this by setting all three inputs of BVAE to the difference view such that the expressive neural representations can also be learned. It is observed that when only the difference view is input, mAP of the model is lower than the result of the input of the left and right hemisphere views (red bar) but miF1 and maF1 is a little bit higher except maF1 in MEMO27 and individual subjects in MEMO80. Model can decode emotional states with relative accuracy when only the difference view is input, which indicates that the discrepancy of the left and right hemispheres indeed provides useful information for emotion decoding. Furthermore, model with all three views as input still has the best performance which means the common component and the discrepancy of the left and right view are complementary and both indispensable for emotion decoding.
\begin{figure*}
	\centering
	\subfigure[ $N_{ROIF}$ in MEMO27.]{\includegraphics[height=1.7in,width=2.2in]{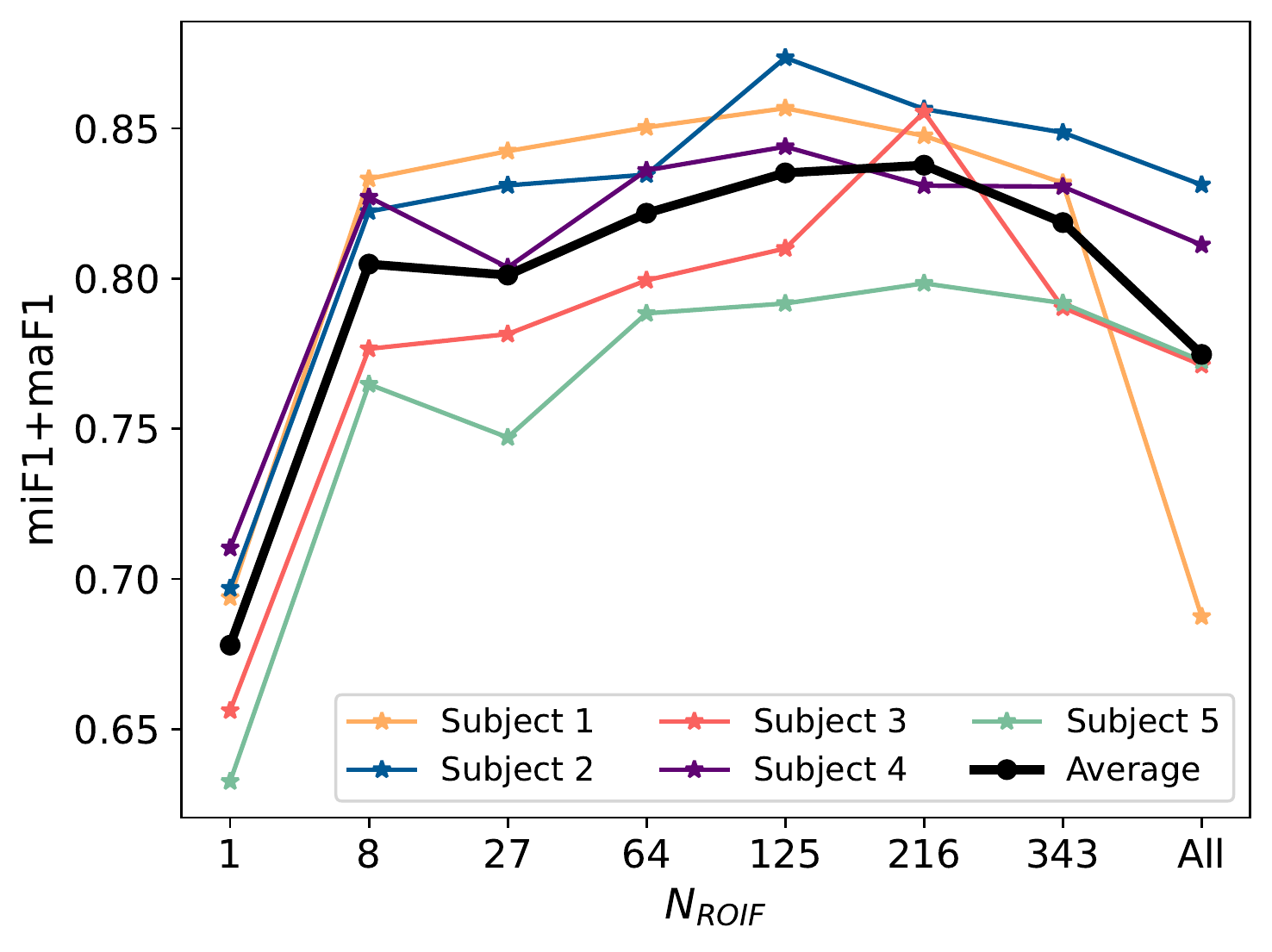}}
	\subfigure[$\lambda_1/\lambda_2$ in MEMO27.]{\includegraphics[height=1.7in,width=2.2in]{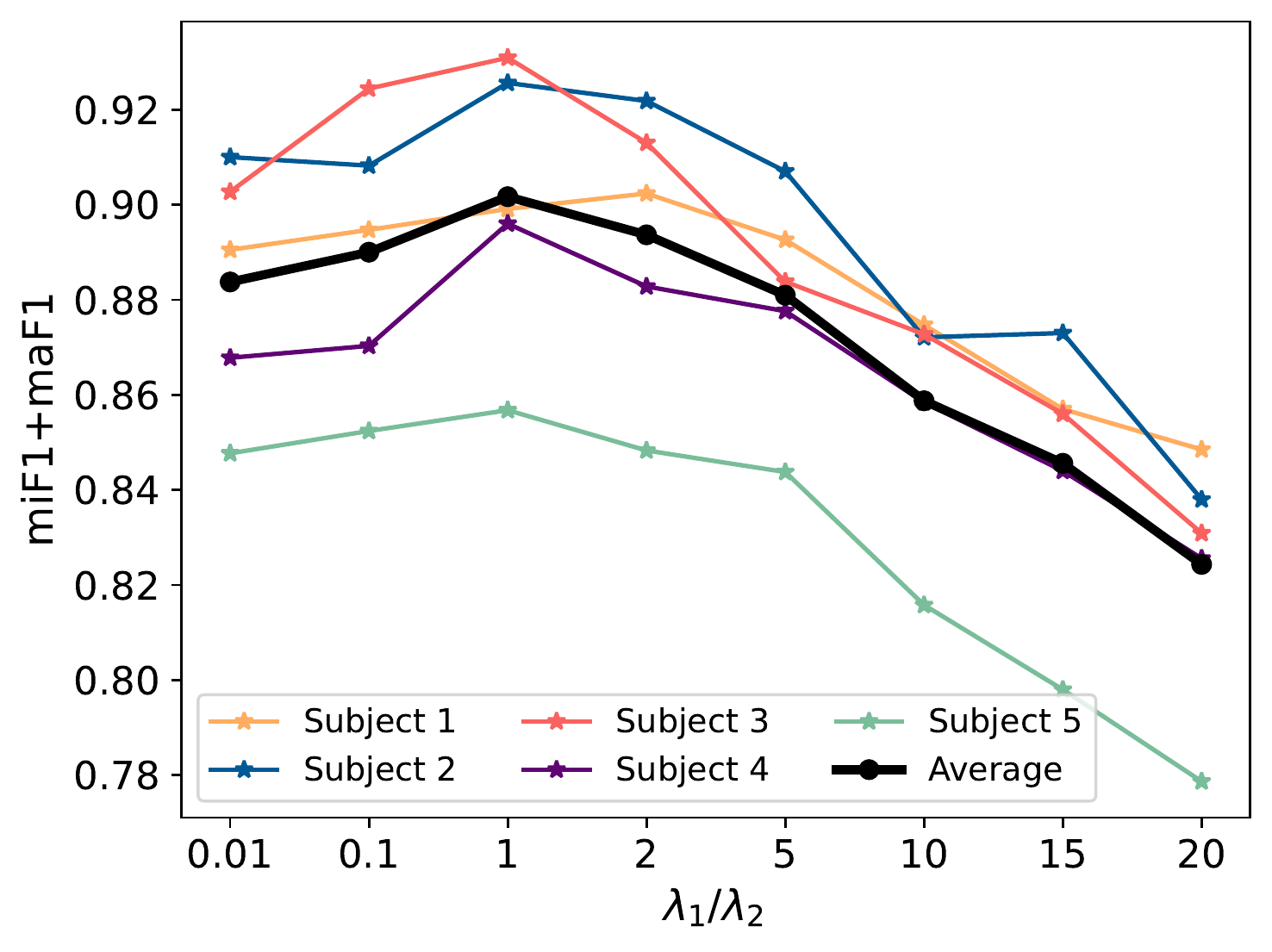}}
	\subfigure[$\lambda_1/\lambda_2$ in MEMO80.]{\includegraphics[height=1.7in,width=2.2in]{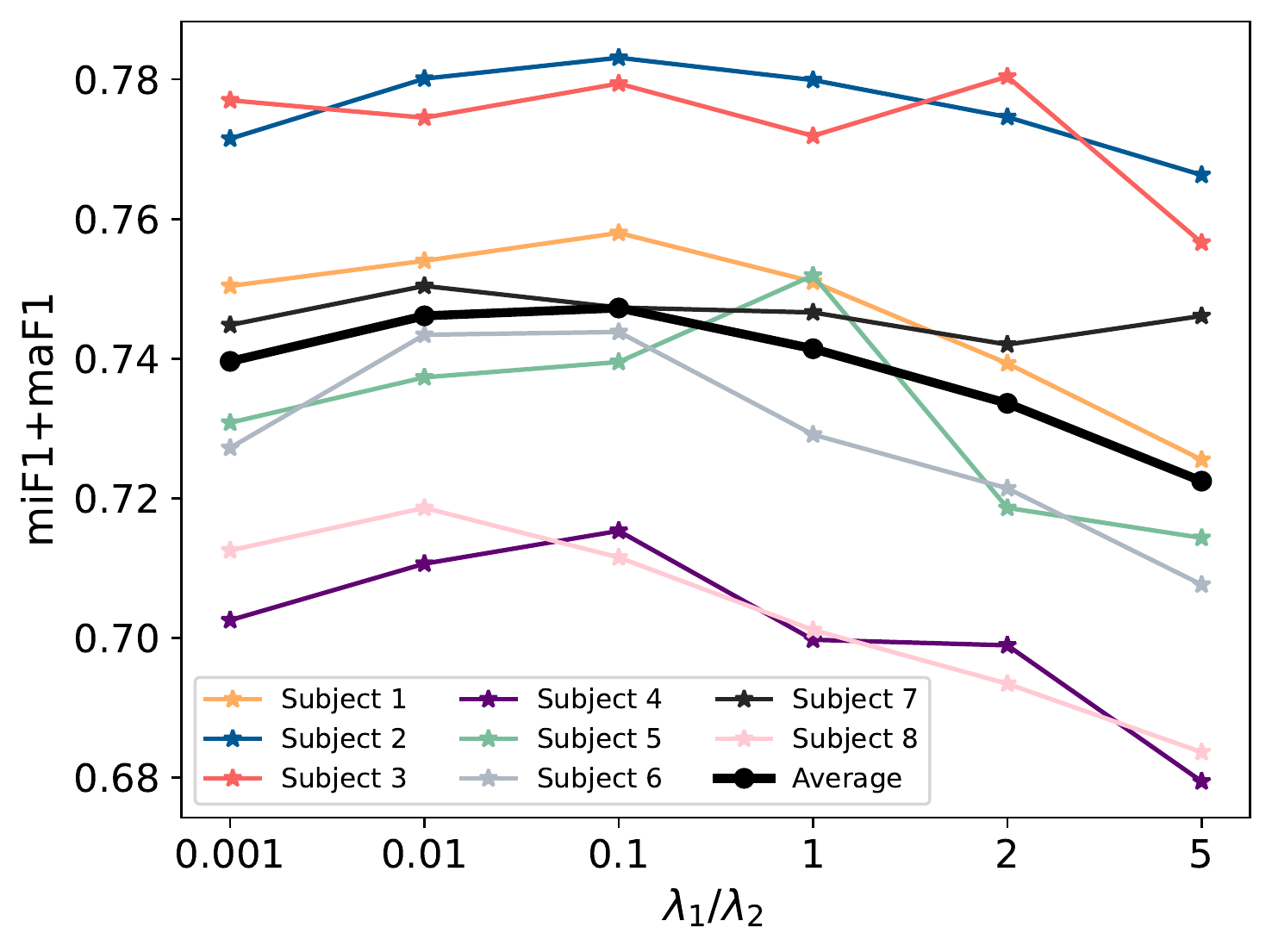}}
	\caption{Performance of ML-BVAE changes as the hyper-parameters $N_{ROIF}$ and $\lambda_1/\lambda_2$ varies in terms of miF1+maF1.}
	\label{ablation_roi}
	\vspace{-0.6cm}
\end{figure*}

\begin{table}[hbt!]
	\centering
	\caption{Comparisons of our method with its variant models for ablation study. $\uparrow$ ($\downarrow$) indicates the larger (smaller) the value, the better the performance (five subjects averaged in MEMO27 and eight subjects averaged in MEMO80). The model performance degradation caused by w/o BVAE and w/o ML-C is marked in terms of the comprehensive metircs.}
	\resizebox{\columnwidth}{!}{
		\begin{tabular}{c|c|ccccc|c} 
			\hline
			\hline
			Dataset & Metrics & w/o ROI-P & w/o BVAE & w/o ML-C&w/o E& w/o M& ML-BVAE\\
			\hline
			\multirow{6}{*}{MEMO27}
			& OneE$\downarrow$ & 0.354 & 0.314 & 0.312&0.309&0.310& \textbf{0.302}\\
			&  RL$\downarrow$  & 0.221  & 0.185  & 0.186 &0.190& 0.187& \textbf{0.186}    \\
			&  miF1$\uparrow$  & 0.462  & 0.380(-0.125)  & 0.461(-0.044)  &0.463&0.455& \textbf{0.505}   \\
			&  maF1$\uparrow$  & 0.313 & 0.230(-0.168)  & 0.291(-0.107)   &0.302&0.281& \textbf{0.398}   \\
			&  e-AP$\uparrow$ & 0.569  & 0.613  & 0.612 &0.611&0.613 & \textbf{0.619}   \\
			&  mAP$\uparrow$  & 0.386  & 0.417(-0.031)  & 0.430(-0.018)  &0.433&0.431& \textbf{0.448}   \\
			\hline
			\hline
			\multirow{6}{*}{MEMO80}
			& OneE$\downarrow$ &0.451 & 0.453 &0.443 &0.481&0.442& \textbf{0.424}\\
			&  RL$\downarrow$  &  0.311 & 0.307  &  0.305 &0.329&0.306& \textbf{0.298}    \\
			&  miF1$\uparrow$  &  0.341 & 0.201(-0.177)  &  0.333(-0.045) &0.374&0.361& \textbf{0.378}    \\
			&  maF1$\uparrow$  &  0.335 & 0.184(-0.185)  &    0.314(-0.055) &0.364&0.351& \textbf{0.369}   \\
			&  e-AP$\uparrow$  & 0.474 &  0.472 &  0.479 &0.453&0.479& \textbf{0.488}   \\
			&  mAP$\uparrow$  &  0.469 &  0.441(-0.030) &  0.460(-0.011) &0.435&0.462& \textbf{0.471}   \\
			\hline
			\hline
	\end{tabular}}
	\label{as}
\end{table}
\subsubsection{Ablation Study}
In order to explore the role of the three key components of our method separately, we design three variant models of ML-BVAE for ablation study. (1) \textbf{w/o ROI-P}: Use whole brain voxel signal instead of ROI pooling features as input to our model; (2) \textbf{w/o BVAE}: Firstly use MLP for single-view features respectively and then concatenate the output of each MLP as multi-view neural representations; (3) \textbf{w/o ML-C}: Directly connect the output of the encoder of BVAE to a linear classifier for prediction without multi-label classification network.

Furthermore, we design another two variant models for more fine-grained ablation study. For verifying the importance of modeling the information intra- and inter- views for neural representations learning, we consider a variant model \textbf{w/o E}: Remove the three extra ELBOs in Eq. \ref{Eq:mmelbo1}- \ref{Eq:mmelbo3} from the loss function. For illustrating the necessary of prior knowledge injection in the label correlation learning, we remove the mask from Eq. \ref{attention_mat} leading to the variant model \textbf{w/o M}.

Results in the two datasets are shown in Table \ref{as} and we have averaged the results of all subjects in each dataset. Results show that all the three key components in ML-BVAE are critical for more accurate emotion decoding from brain activity. Furthermore, the extra three ELBOs are helpful to neural representations learning and the prior knowledge mask facilitates label correlation learning, which can both contribute to multi-label emotion decoding.

Besides, it is necessary to analyze whether the high performance is attributed to the discriminative component with huge parameter space, so as to exclude the posterior collapse phenomenon. We mark the model performance degradation caused by the lack of BVAE (the generative component) and ML-C (the discriminative component) in the table in terms of the comprehensive metrics miF1, maF1 and mAP. We can observe that, compared with w/o ML-C, w/o BVAE has a greater impact on model performance. That is to say, the high performance should be attributed more to the expressive neural representations learning of BVAE.

\subsubsection{Parameter Sensitivity}
The two most important hyperparameters in our experiment are the number of features per ROI $N_{ROIF}$ and the trade-off parameter $\lambda_1/\lambda_2$ in Eq. ($\ref{loss_function}$). We design experiments by varying one parameter while fixing another parameter. We utilize the criterion miF1+maF1 \cite{yeh2017learning} for parameter selection.

Fig. \ref{ablation_roi} (a) gives an illustrative example of how the performance of ML-BVAE changes when $N_{ROIF}$ changes in MEMO27 \footnote{Since the number of voxels for each ROI in the MEMO80 is small, we select 8 by default.}. We find that both global average pooling of the voxels in ROI and utilizing all voxels are not suitable. In general, the performance of ML-BVAE increases first and then decreases with the increase of this parameter. Therefore we select the parameter which can reach the optimal performance in terms of this criterion for each subject.

Fig. \ref{ablation_roi} (b) and (c) show the performance of ML-BVAE changes when the trade-off parameter $\lambda_1/\lambda_2$ varies. In MEMO27, generally speaking, the model performance increases first and then falls down as the parameter increases. It is noticeable that the performance can be significantly degraded when the parameter becomes too large. Similar conclusions can be drawn in MEMO80.

\vspace{-0.5cm}
\section{Conclusion}
We have proposed a hybrid model ML-BVAE for fine-grained multi-view multi-label emotion decoding from visually evoked brain activity. The proposed method can be divided into three key components. Firstly, we employed ROI pooling for fMRI dimensionality reduction which can alleviate overfitting. Secondly, BVAE was used for multi-view neural representations learning and we also took the discrepancy of bi-hemisphere into consideration. At last, a multi-label classification network was implemented for emotion-specific representation learning and modeling the dependency of emotion labels. Our method can extract expressive neural representations for accurate multi-label emotion decoding up to 80 fine-grained emotion categories. We leveraged two fine-grained multi-label emotion decoding benchmark datasets, and comprehensive experiments on them have confirmed the superiority of the proposed method. In the future, decoding emotional states from noisy labels may be a promising research direction as there are some emotion categories whose ratings have low consistency among annotators.
\vspace{-0.5cm}
\section*{Acknowledgments}
This work was supported the National Key Research and Development Program of China under Grant 2021ZD0201503; in part by the National Natural Science Foundation of China under Grant 62206284, 61976209, 61906188; in part by the CAS International Collaboration Key Project under Grant 173211KYSB20190024; in part by Beijing Natural Science Foundation under Grant J210010 and Grant 7222311 and in part by the Strategic Priority Research Program of CAS under Grant XDB32040200. We thank Tomoyasu Horikawa and Naoko Koide-Majima for providing the necessary data for our research.

\appendix
\section*{A. Proof of the distribution of a product of Gaussian experts}
\begin{prop}
	Give a finite number $N$ of multi-dimensional Gaussian distributions $p_{i}(\mathbf{x})$ with mean $\boldsymbol{\mu}_{i}$ and covariance $\mathbf{\Sigma}_{i}$, in which $i = 1, ..., N$, the product $\prod_{i=1}^{N} p_{i}(\mathbf{x})$ is Gaussian with mean $(\sum_{i=1}^{N} \mathbf{\Sigma}_{i}^{-1}\boldsymbol{\mu}_{i})(\sum_{i=1}^{N} \mathbf{\Sigma}_{i}^{-1})^{-1}$ and covariance $(\sum_{i=1}^{N} \mathbf{\Sigma}_{i}^{-1})^{-1}$.
\end{prop}
\begin{proof}\let\qed\relax The canonical form of the probability density of a Gaussian distribution is $K\exp\{\boldsymbol{\eta}^{T}\mathbf{x} - \frac{1}{2}\mathbf{x}^{T}\mathbf{\Lambda} \mathbf{x}\}$ where $K$ is a normalizing constant, $\mathbf{\Lambda} =\mathbf{\Sigma}^{-1}$, $\boldsymbol{\eta} = \mathbf{\Sigma}^{-1}\boldsymbol{\mu}$. Then we can get the probability density of product of $N$ Gaussian distributions $\prod_{i=1}^{N} p_i \propto \exp\{(\sum_{i=1}^{N} \boldsymbol{\eta}_i)^{T}\mathbf{x} - \frac{1}{2}\mathbf{x}^{T}(\sum_{i=1}^{N} \mathbf{\Lambda}_{i})\mathbf{x}\}$. Therefore, this product itself has the form of Gaussian distribution with $\boldsymbol{\eta} = \sum_{i=1}^{N} \boldsymbol{\eta}_i$ and $\mathbf{\Lambda} = \sum_{i=1}^{N}\mathbf{\Lambda}_i$. Converting back from canonical form, the product Gaussian has mean $\boldsymbol{\mu} = (\sum_{i=1}^{N} \mathbf{\Sigma}_{i}^{-1}\boldsymbol{\mu}_{i})(\sum_{i=1}^{N} \mathbf{\Sigma}_{i}^{-1})^{-1}$ and covariance $(\sum_{i=1}^{N} \mathbf{\Sigma}_{i}^{-1})^{-1}$.
\end{proof}
\section*{B. Bar graphs of comparisons of AP and mAP in $\%$ of ML-BVAE and compared methods}
\begin{figure*}
	\centering
	\includegraphics[height=1.6in,width=7.1in]{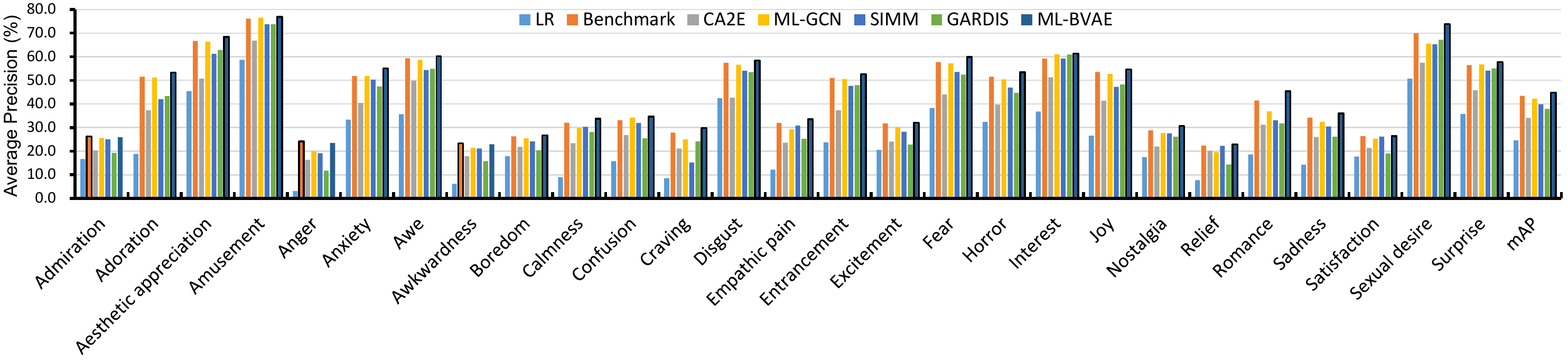}
	\caption{Comparisons of AP and mAP in $\%$ of ML-BVAE and compared methods in MEMO27 (five subjects averaged). The best performance is marked with a black outline.}
	\label{map27_bar}
\end{figure*}
\begin{figure*}
	\centering
	\subfigure {\includegraphics[height=1.6in,width=7.1in]{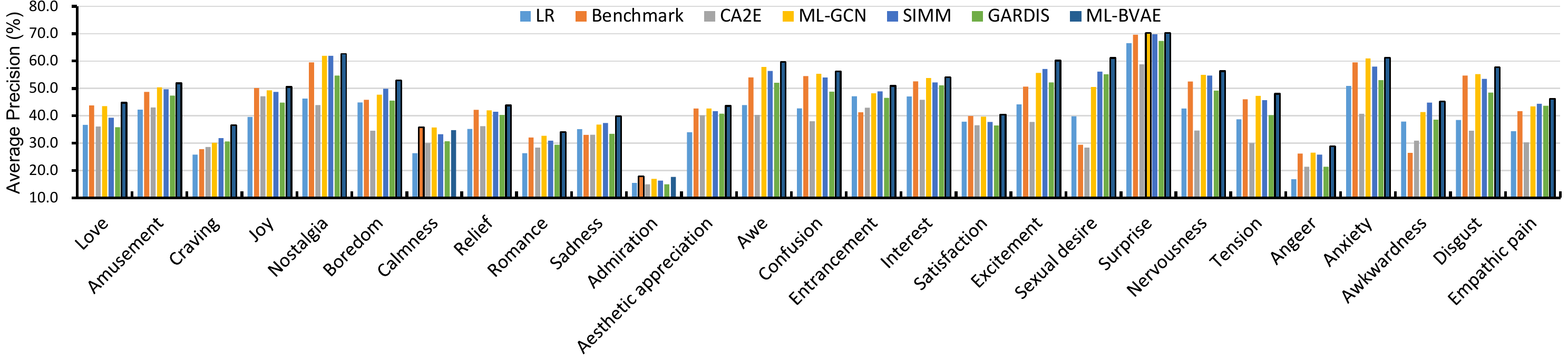}}
	\subfigure {\includegraphics[height=1.6in,width=7.1in]{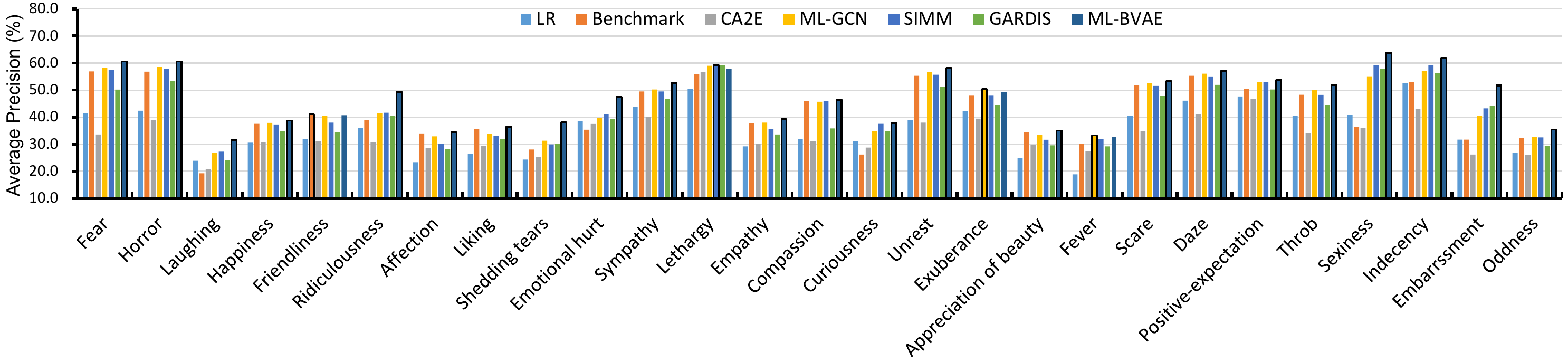}}
	\subfigure {\includegraphics[height=1.6in,width=7.1in]{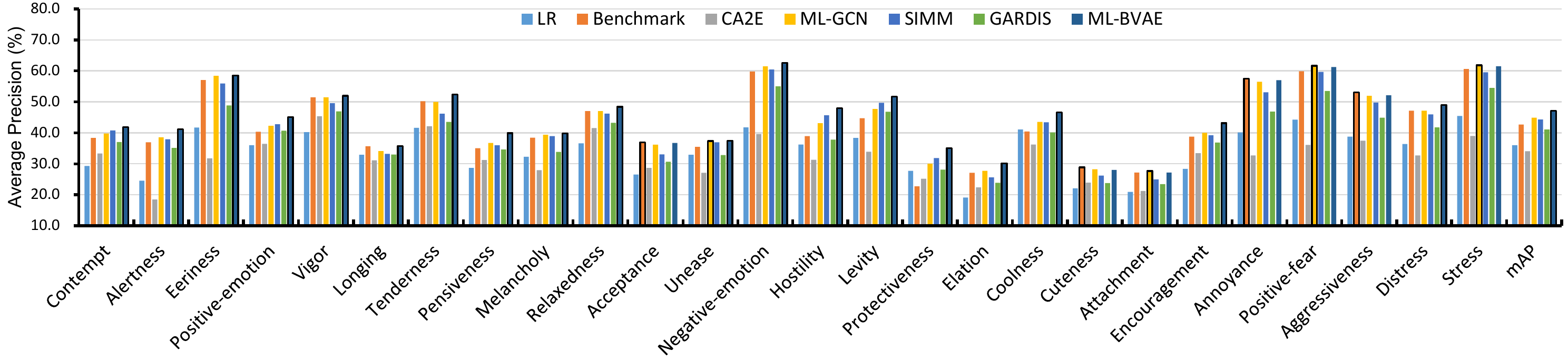}}
	\caption{Comparisons of AP and mAP in $\%$ of ML-BVAE and compared methods in MEMO80 (eight subjects averaged). The best performance is marked with a black outline.}
	\label{map80_bar}
\end{figure*}

\newpage
\bibliographystyle{IEEEtran}
\bibliography{output}

\vfill

\end{document}